\documentclass[journal]{IEEEtran}
\usepackage{amsmath,amsfonts}
\usepackage{array}
\usepackage[caption=false,font=normalsize,labelfont=sf,textfont=sf]{subfig}
\usepackage{textcomp}
\usepackage{stfloats}
\usepackage{url}
\usepackage{cite}
\usepackage{verbatim}
\usepackage{algorithm}
\usepackage{tabularx}
\usepackage{algpseudocode}
\usepackage[english]{babel}
\usepackage{amsthm}

\newtheorem{lemma}{\quad \textbf{Lemma}}

\usepackage{graphicx}
\def\BibTeX{{\rm B\kern-.05em{\sc i\kern-.025em b}\kern-.08em
    T\kern-.1667em\lower.7ex\hbox{E}\kern-.125emX}}
\usepackage{balance}
\begin{document}
\title{Joint Vehicle Connection and Beamforming Optimziation in Digital Twin Assisted Integrated Sensing and Communication Vehicular Networks}
\author{Weihang Ding~\IEEEmembership{Student Member, IEEE}, Zhaohui Yang~\IEEEmembership{Member, IEEE}, Mingzhe Chen~\IEEEmembership{Member, IEEE},\\ Yuchen Liu~\IEEEmembership{Member, IEEE} and Mohammad Shikh-Bahaei~\IEEEmembership{Senior Member, IEEE}

\thanks{W. Ding and M. Shikh-Bahaei are with Department of Engineering, King's College London, London, UK. (email: weihang.ding@kcl.ac.uk; m.sbahaei@kcl.ac.uk)}
\thanks{Z. Yang is with the College of Information Science and Electronic Engineering, Zhejiang University, Hangzhou, Zhejiang 310027, China, and Zhejiang Provincial Key Lab of Information Processing, Communication and Networking (IPCAN), Hangzhou, Zhejiang, 310007, China (Email: yang\_zhaohui@zju.edu.cn). }
\thanks{M. Chen is with the Department of Electrical and Computer Engineering and Frost Institute for Data Science and Computing, University of Miami, Coral Gables, FL 33146 USA (Email: \protect\url{ mingzhe.chen@miami.edu)}.} 
\thanks{Y. Liu is with the Department of Computer Science, North Carolina State University, Raleigh, NC, 27695, USA (Email: yuchen.liu@ncsu.edu).}
\vspace{-1.5em}

\thanks{This work is supported in part by National Key R\&D Program of China (Grant No. 2023YFB2904804), Young Elite Scientists Sponsorship Program by CAST 2023QNRC001,  Zhejiang Key R\&D Program under Grant 2023C01021, the Fundamental Research Funds for the Central Universities, K2023QA0AL02}
\vspace{-1.5em}
\thanks{This work is supported in part by the U.S. National Science Foundation under Grants CNS-2312138, CNS-2312139, and CNS-2332834.}

}


\maketitle

\begin{abstract}

This paper introduces an approach to harness digital twin (DT) technology in the realm of integrated sensing and communications (ISAC) in the sixth-generation (6G) Internet-of-everything (IoE) applications. We consider moving targets in a vehicular network and use DT to track and predict the motion of the vehicles. After predicting the location of the vehicle at the next time slot, the DT designs the assignment and beamforming for each vehicle. 
The real time sensing information is then utilized to update and refine the DT, enabling further processing and decision-making. In the DT, an extended Kalman filter (EKF) is used for precise motion prediction. This model incorporates a dynamic Kalman gain, which is updated at each time slot based on the received echo signals. The state representation encompasses both vehicle motion information and the error matrix, with the posterior Cramér-Rao bound (PCRB) employed to assess sensing accuracy. We consider a network with two roadside units (RSUs), and the vehicles need to be allocated to one of them. To optimize the overall transmission rate while maintaining an acceptable sensing accuracy, an optimization problem is formulated. Since it is generally hard to solve the original problem, Lagrange multipliers and fractional programming are employed to simplify this optimization problem. To solve the simplified problem, this paper introduces both greedy and heuristic algorithms through optimizing both vehicle assignments and predictive beamforming. The optimized results are then transferred back to the real space for ISAC applications. Recognizing the computational complexity of the greedy and heuristic algorithms, a bidirectional long short-term memory (LSTM)-based recurrent neural network (RNN) is proposed for efficient beamforming design within the DT. Simulation results demonstrate the effectiveness of the DT-based ISAC network. Notably, the LSTM-based RNN method achieves similar transmission rates as the heuristic algorithm but with significantly reduced computational complexity.


\end{abstract}

\begin{IEEEkeywords}
Integrated sensing and communication, vehicular network, digital twin.
\end{IEEEkeywords}

\section{Introduction}

\IEEEPARstart{D}{igital} twins (DTs) that serve as virtual representations of physical objects, systems, or processes, offering detailed and dynamic simulations of their real-world counterparts, has gained significant attention from both academia and industry. 
Specifically, through constantly learning and updating from the real space, DTs accurately characterize the working conditions and locations of physical entities, enabling precise predictions of future events. Therefore, DTs have been applied for widespread applications in edge computing \cite{JiangLi, DaiYueyue, Do-Duy}, the Internet of things (IoT) \cite{zhangPeiying, WangHuan, LuYunlong, Chukhno}, cyber-physical systems \cite{YunSeongjin, LinWD, XieGuoqi}, and vehicle networks.

DT offers a dynamic virtual representation of vehicular network systems, mirroring the physical state, processes, and real-time systems of vehicles or networks. With various sensors installed on the vehicle and network infrastructure, as well as broader contextual information such as traffic conditions and environmental factors, DTs are capable of simulating, predicting with data, and optimizing the performance and maintenance of the physical counterparts \cite{Zhangke}. DT has been demonstrated to be effective in traffic management and optimization \cite{FuYanfang, DongJianghong}. By simulating traffic flows and vehicle interactions within a DT of a vehicular network, cities and organizations can optimize traffic patterns, reduce congestion, and enhance safety. Additionally, by simulating various crash scenarios and cyber-attack simulations, DT contributes to improving vehicle safety features and enhancing resilience against attacks on vehicular networks \cite{Khan, Shadrin}. Furthermore, DT can play a crucial role in monitoring the real-time health and performance of vehicles, predicting maintenance needs \cite{Dimitrova, WuXingtang}, and forecasting battery charging requirements \cite{YuGang}. Given these applications and the rising trend of developing autonomous vehicles, the collaboration between DTs and integrated sensing and communication (ISAC) in the sixth-generation (6G) IoE holds great promise \cite{yang2023joint, yang2024optimizing}.


The integration of sensor technologies with advanced communication (ISAC) systems \cite{Chiriyath} has been attracted great attention in recent years. To realize ISAC, the authors in \cite{Ghatak} proposed a millimeter wave (mmWave) system to support both positioning and downlink broadband services and analyzed the trade-off between positioning accuracy and communication rate. The authors in \cite{Destino} explored this trade-off in a similar system, by splitting the overall time for beam alignment and data transmission respectively. The optimal overhead in a multi-user scenario was designed in \cite{Kumar}, balancing the communication and sensing performances. The above works \cite{Chiriyath, Ghatak, Destino, Kumar} all focused on a scenario where the mobility of the users is limited. For high-speed vehicles, for example, vehicles that are driven on highways, the delay induced by the overhead is intolerable since it will lead to outdated motion prediction. To reduce the delay, dual-functional radars are proposed, that is, both radar sensing and communication are achieved with one dual-functional device and the same signal \cite{LiuFan}. In \cite{Xu}, a single radar waveform has been used for both sensing and carrying information. This scheme is a low-data-rate scheme because the information bits are transmitted by selecting either the down-chirp waveform or the up-chirp waveform. Due to inherent correlations, the sensing information can be further used for channel estimation \cite{ChenXu, LinBangjiang, Hossein}.

The innovation of the fifth-generation (5G) technology, which exploits both massive multiple-input multiple-output (mMIMO) antenna arrays and the mmWave spectrum, provides the opportunity for reliable ISAC systems \cite{Heath}. Compared to MIMO, mMIMO deploys a much larger array. It has been shown in \cite{Marzetta} that if the transmitter uses an infinite number of antennas to serve a limited number of users, fast-fading tends to vanish and the channels tend to be orthogonal. In other words, the system can be viewed as interference-free when the vehicular network is relatively sparse. The mMIMO technique can be used to compensate for the excessive path-loss of the mmWave signals and formulate ``pencil-like" beams to concentrate the signal power in the desired direction \cite{F.Liu}. The beam-squint effect, induced by the increasing number of antennas, is harnessed in conjunction with the beam-split effect for Joint Channel Assignment and Scheduling (JCAS) in the mmWave/terahertz (THz) band \cite{GaoFeifei}. Therefore, mMIMO and mmWave are an excellent combination to achieve high resolution and high rate and can support ISAC in the sixth-generation (6G) network \cite{FangXinran, QiQiao}. In addition, it is also possible to apply positioning reference signals (PRSs) in 5G for sensing, positioning, and communication \cite{WeiZhiqing}.

Previous research primarily focused on improving the performance of ISAC systems by designing dual-functional transmit signal waveforms. In \cite{Yongjun}, the authors designed adaptive integrated waveforms that improved both the mutual information of the impulse response and the data rate. In \cite{Huang}, the authors proposed a scheme for embedding digital information into radar signals and designed a low-complexity receiver to recover the information. In \cite{MaoTianqi}, the authors designed a data-embedded multi-subband quasi-perfect waveform for ISAC in mmWave and low THz band. Besides, phase coding has been applied for ISAC signal design in \cite{WeiZhiqing2} to reduce the effect of noise. More recent works have focused on beamforming design to enhance performance. Extended Kalman filter (EKF) and learning algorithms were respectively applied in \cite{LiuFan} and \cite{LiuChang} to jointly optimize sensing and communication performance by exploiting the Cramer-Rao lower bound (CRB). In \cite{YuanWeijie}, the authors proposed a message-passing algorithm for estimating vehicle states and applied a Bayesian approach to analyze information content propagating on the factor graph. In \cite{LiuXiang}, the transmitter utilizes the jointly precoded communication and radar waveforms to design dual-functional transmit beamforming, which optimizes both functions simultaneously. Furthermore, the security challenges in ISAC were studied in \cite{Chang, ChuJinjin, Gunlu, YangYinchao}.

Several works have jointly considered DT and ISAC for IoE applications. In \cite{Gongyu}, the authors focused on resource allocation for ISAC in DT without designing an explicit DT model. Cui et al. \cite{CuiYuanhao} proposed an ISAC waveform design method without accounting for the movement of objects. In \cite{LiBin}, the authors used DT to make intelligent offloading decisions in an unmanned aerial vehicle (UAV)-assisted ISAC network but primarily concentrated on stationary ground targets. Mu et al. \cite{MuJunsheng} proposed a Federated Learning (FL)-empowered DT-based communication-assisted sensing network with synthetic aperture radar (SAR) distributed on aircraft. In \cite{Román-García}, the authors proposed a smart building DT by sensing power consumption and acting according to environmental factors. However, \cite{Román-García} lacks mathematical analysis. Liu et al. \cite{LiuJingxian} proposed a DT-based method to intelligently predict the state and trajectory of moving targets, but the targets are maneuvering, and communication is not involved. This paper involves the design of a comprehensive DT for modeling and predicting vehicle movements, then collaboratively assigning vehicles to two roadside units (RSUs) and designing predictive beamforming for the next time slot. A major advantage of applying DT in this work is that it provides a centralized framework for aggregating data from multiple RSUs. This centralized approach to data handling enables more comprehensive analytics, better decision-making, and a unified view of the entire system, as opposed to the isolated sensing information provided by individual RSUs. The DT can make real-time decisions and perform real-time simulations in response to the real-time information collected from the physical environment. Furthermore, considering the allocation of vehicles to different RSUs, it is essential that the information collected from various RSUs be managed jointly. Therefore, decisions must be made using DT rather than by each RSU independently. In our work, we only consider the scenario with two RSUs, but there might be multiple RSUs in reality. However, due to path loss effects, the performance gains obtained from RSU selection are only significant when the distances between the vehicle and the RSUs are similar. Since the coverages of the RSUs can be seen as circles, we only consider the vehicles located near the line connecting the intersection points of two circles. In a system with multiple RSUs, we can simplify the problem by decomposing it into problems between each pair of RSUs. EKF is applied to track the movement of the vehicles in DT. The real-time Kalman gain is updated based on the difference between the received and expected echo signals. Subsequently, the predicted states are corrected using the Kalman gain and the echo signal. Our main focus is on maximizing the transmission rate, with sensing accuracy serving as a constraint. To address this optimization problem, we introduce both a greedy algorithm and a heuristic algorithm. Furthermore, we propose a bidirectional long short-term memory (LSTM)-based recurrent neural network (RNN) to predict the beamforming of the vehicles. The key contributions of this paper are outlined as follows:

\begin{itemize}
    \item A DT is employed for real-time modeling of vehicles in an ISAC-based vehicular network. The motion model is based on an EKF, and the Kalman gain is updated at each time slot based on the received echo signal to accurately track the real-time location and movement of vehicles. The state of the vehicle includes not only the motion information of the vehicle but also the error matrix, and the posterior Cramér-Rao Bound (PCRB) is used to evaluate the sensing accuracy of the vehicles.
    \item After the DT obtains the predicted locations of the vehicles, we formulate an optimization problem to optimize the communication performance of the system. Sensing accuracy is considered as a constraint to ensure an accurate DT for the real network. We consider a vehicular network with two RSUs, so the vehicles need to be assigned to one of them. This optimization problem is simplified using Lagrange multipliers and fractional programming. To address this challenge, we propose both a greedy algorithm and a heuristic algorithm. These algorithms are designed to jointly optimize vehicle assignments and predictive beamforming, aiming to maximize the overall transmission rate.
    \item Due to the high computational complexity of both the greedy algorithm and the heuristic algorithm, we introduce an efficient approach using a bidirectional LSTM-based method to design beamforming for each vehicle in the DT. The bidirectional LSTM is employed to account for the high correlation among the beamformings of adjacent vehicles. Simulation results demonstrate that the LSTM-based method achieves a transmission rate only slightly lower than the heuristic algorithm while significantly reducing computational complexity.
\end{itemize}

The subsequent sections of this paper are organized as follows. In Section II, we present the system model, including the radar sensing model and the communication model. Section III introduces the proposed DT framework. The optimization problem is formulated in Section IV and solved in Section V. An alternative method utilizing LSTM networks is discussed in Section VI. Section VII presents all relevant simulation results. Finally, Section VIII concludes the whole paper.

\section{System Model}

\begin{figure*}[!t]
    \centering
    \includegraphics[width=1.75\columnwidth]{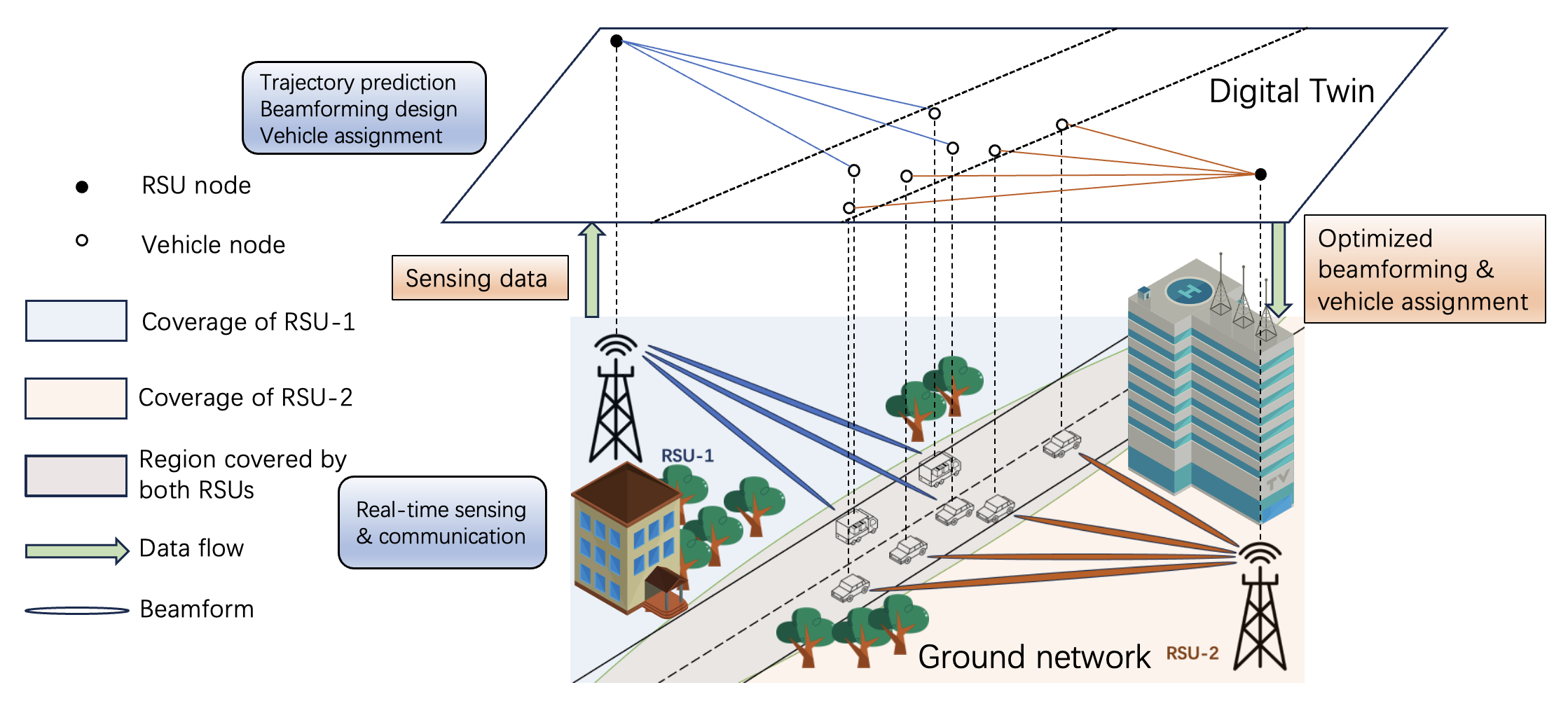}
    \caption{The considered vehicle network with 2 RSUs and $K$ vehicles. }
    \label{model}
\end{figure*}

We consider a vehicular network that consists of two RSUs equipped with mMIMO antennas\footnote{One can easily extend the considered model with two RSUs to a network with multiple RSUs since we have considered the overlap caused by two RSUs.} serving $K$ vehicles, as shown in Fig. \ref{model}. At each time slot, the RSUs simultaneously transmit information to the vehicles and sense the position of these vehicles. The sensed information is then forwarded to the DT where a two-dimensional projection of the physical network is constructed. Using matched filtering and EKF, the DT continuously tracks the vehicles, predicts their movements, and optimizes both vehicle assignments and beamforming for the next time slot to achieve high-capacity ISAC. Next, we first introduce the radar sensing model. Then, we explain the communication model including the channel gain and signal-to-interference-plus-noise
ratio (SINR) expression.  

\subsection{Radar Sensing}

Both RSUs are equipped with a mmWave-band mMIMO uniform linear array (ULA) consisting of $N_t$ transmit antennas and $N_r$ receive antennas. Each vehicle is equipped with only one antenna. Since there are $K$ vehicles that need to be served in the system, at the $n$-th time slot, the $i$-th RSU transmits a multi-dimensional multi-beam ISAC signal. Based on the information measured from the echo signal received at the $(n-1)$-th time slot, the beamforming matrix of the $i$-th RSU at time slot $n$ can be expressed as
\begin{equation}
    \mathbf{F}_{i,n} = [\mathbf{f}_{[i,1],n}\xi_{[i,1],n},  \dots, \mathbf{f}_{[i, K],n}\xi_{[i,K],n}] \in \mathbb{C}^{N_t\times K},
\end{equation}
where $\mathbf{f}_{[i,k],n}\in \mathbf{C}^{N_t\times 1}$ is the beamforming vector for the $k$-th vehicle, and $\xi_{[i,k],n}\in\{0,1\}$ indicates whether vehicle $k$ is connected to RSU $i$ at time slot $n$. Thus, the transmitted signal at RSU $i$ can be expressed as
\begin{equation}
    \tilde{\mathbf{s}}_{i,n}(t) = \mathbf{F}_{i,n}\mathbf{s}_{i,n}(t)\in \mathbb{C}^{N_t \times 1},
\end{equation}
where $\mathbf{s}_{i,n}(t)=[s_{[i,1],n}(t), s_{[i,2],n}(t), \dots, s_{[i,K],n}(t)]^{T}\in \mathbb{C}^{K \times 1}$ is the ISAC signal.

Since the signal is omnidirectional, the signal received by each RSU contains the echoes reflected by all vehicles within the coverage of the RSU, regardless of whether they are connected to the RSU or not. Thus, the reflected echo signal received at the $i$-th RSU at time slot $n$ can be expressed as
\begin{equation}
\begin{aligned}
    \mathbf{r}_{i,n}(t) & = \kappa \sum_{k=1}^{K} \beta_{[i,k],n}e^{j2\pi\mu_{[i,k],n}t}\mathbf{b}(\varphi_{[i,k],n})\mathbf{a}^{H}(\varphi_{[i,k],n})\\
    &\cdot \tilde{s}_{i,n}(t-\nu_{[i,k],n})+\mathbf{z}_{i,n}(t),
\end{aligned}
\end{equation}
where $\kappa=\sqrt{N_tN_r}$ is the antenna gain, $\mu_{[i,k],n}$ and $\nu_{[i,k],n}$ are respectively the Doppler frequency and the time delay with respect to vehicle $k$ and RSU $i$ at time slot $n$, $\mathbf{a}(\varphi_{[i,k],n})\in \mathbb{C}^{N_t\times 1}$ and $\mathbf{b}(\varphi_{[i,k],n})\in \mathbb{C}^{N_r\times 1}$ are respectively the transmit and receive steering vectors, with $\varphi_{[i,k],n}=\cos{\theta_{[i,k],n}}$, $\mathbf{z}_{i,n}(t)\in \mathbb{C}^{N_r\times 1}$ denotes the noise vector at RSU $i$ with $\sigma_e^2$ being the variance of each element of $\mathbf{z}_{i,n}(t)$, and $\beta_{[i,k],n}=\frac{\varrho}{2d_{[i,k],n}}$ is the reflection coefficient with $d_{[i,k],n}$ being the distance between vehicle $k$ and RSU $i$ at time slot $n$ and $\varrho$ represents the complex fading coefficient that depends on the radar cross-section. The transmit and receive steering vectors $\mathbf{a}(\varphi_{[i,k],n})$ and $\mathbf{b}(\varphi_{[i,k],n})$ can be respectively expressed as
\begin{equation}
    \mathbf{a}(\varphi_{[i,k],n}) = \sqrt{\frac{1}{N_t}}[1, e^{-j\pi \varphi_{[i,k],n}}, \dots, e^{-j\pi (N_t-1)\varphi_{[i,k],n}}]^T,
\end{equation}
\begin{equation}
    \mathbf{b}(\varphi_{[i,k],n}) = \sqrt{\frac{1}{N_r}}[1, e^{-j\pi \varphi_{[i,k],n}}, \dots, e^{-j\pi (N_r-1)\varphi_{[i,k],n}}]^T.
\end{equation}

The steering vectors corresponding to different vehicles in the network are asymptotically orthogonal. However, we cannot assume that $|\mathbf{a}^H(\varphi_{[i,k],n})\mathbf{a}(\varphi_{[i,k'],n})|=0$ for all $k\neq k'$ in a dense-traffic vehicular network, as it is more likely that multiple vehicles are at similar angles relative to the RSU. Since steering vectors are asymptotically orthogonal to each other, only echoes from vehicles at similar angles will interfere with each other. At the $i$-th RSU, the echoes from the $k$-th vehicle can be extracted using the following equation:
\begin{equation}
\begin{aligned}
    \mathbf{r}_{[i,k],n}(t) &=\kappa\beta_{[i,k],n}e^{j2\pi\mu_{[i,k],n}t}\mathbf{b}(\varphi_{[i,k],n})\mathbf{a}^{H}(\varphi_{[i,k],n})\\
    & \cdot \mathbf{f}_{[i,k],n}\xi_{[i,k],n}s_{[i,k],n}(t-\nu_{[i,k],n})+\mathbf{z}_{[i,k],n}(t),
\end{aligned}
\end{equation}
where $\mathbf{z}_{[i,k],n}(t)\in \mathbb{C}^{N_r\times 1}$ is the summation of the noise and interference caused by the echoes of other vehicles. With the existence of inter-beam interference, $\mathbf{z}_{[i,k],n}(t)$ can be written as
\begin{equation}
\begin{aligned}
    \mathbf{z}_{[i,k],n}(t) &= \kappa \sum_{m\ne k}^{K} \beta_{[i,m],n}e^{j2\pi\mu_{[i,m],n}t}\mathbf{b}(\varphi_{[i,k],n})\mathbf{a}^{H}(\varphi_{[i,k],n}) \\
    & \cdot \mathbf{f}_{[i,m],n}\xi_{[i,k],n}s_{[i,k],n}(t-\nu_{[i,m],n})+\mathbf{z}_{i,n}(t).
\end{aligned}
\end{equation}

\subsection{Communication Model}
Different from previous works \cite{1,2,3,4,5,6,7,8,9}, in this work, the RSUs transmit control information, DT-related data, and personal data traffic to the vehicles. The vehicles are equipped with a single antenna, thus the communication between the RSU and the vehicles forms a multiple-input and single-output (MISO) system. At time slot $n$, the signal received by vehicle $k$ from RSU $i$ can be expressed as
\begin{equation}
\begin{aligned}
    c_{[i,k],n}(t) = &\kappa'\sqrt{\alpha_{[i,k],n}}e^{j2\pi\mu_{[i,k],n}t}\mathbf{a}^{H}(\varphi_{[i,k],n})\tilde{\mathbf{s}}_{i,n}(t)+z_{c}(t),
\end{aligned}
\end{equation}
where $\kappa'=\sqrt{N_t}$ represents the transmit antenna gain, and $z_{c}(t)\sim\mathcal{N}(0,\sigma_{c}^2)$ denotes the communication noise. Finally, $\alpha_{[i,k],n}$ is the path-loss coefficient, which can be computed by
\begin{equation}
    \alpha_{[i,k],n} = \tilde{\alpha} d_{[i,k],n}^{-2},
\end{equation}
where $\tilde{\alpha}$ is the channel power gain at the reference distance $d_0 = 1$m. Ideally, the beamforming matrices of different vehicles are asymptotically orthogonal, which yields $\mathbf{a}^{H}(\varphi_{[i,k],n})\tilde{\mathbf{s}}_{i,n}(t) \approx s_{[i,k],n}(t)$. However, interference between vehicles is often unavoidable in a dense network, especially for vehicles that are overlapping or at similar angles relative to a RSU. The SINR
at vehicle $k$ can be written as
\begin{equation}
\begin{aligned}
    &\text{SINR}_{[i,k],n} = \\
    &\frac{\kappa'^2|\alpha_{[i,k],n}||\mathbf{a}^{H}(\varphi_{[i,k],n})\mathbf{f}_{[i,k],n}|^2\xi_{[i,k],n}}{\sum_{m\neq k}^{K}\kappa'^2|\alpha_{[i,k],n}||\mathbf{a}^{H}(\varphi_{[i,k],n})\mathbf{f}_{[i,m],n}|^2\xi_{[i,m],n}+\sigma_c^2}.
\end{aligned}
\end{equation}

\section{DT-based State Transition and Prediction}
Next, we introduce the construction of the DT with the sensed information.
Sensing information of the physical network is forwarded to the DT layer for data processing. After matched filtering, the state of the vehicle at the next time slot can be predicted. Additionally, optimal vehicle assignment and beamforming can be designed based on the predicted states of the vehicles. For vehicle $k$ with respect to RSU $i$ at time slot $n$, we define $\mathbf{x}_{[i,k],n}=[\varphi_{[i,k],n}, d_{[i,k],n}, \dot{v}_{[i,k],n}]^T \in \mathbb{C}^{3 \times 1}$ as its state, and $\mathbf{y}_{[i,k],n}=[\tilde{\mathbf{r}}_{[i,k],n}, \tilde{\nu}_{[i,k],n}, \tilde{\mu}_{[i,k],n}]^T \in \mathbb{C}^{(N_t+2) \times 1}$ as the measured parameters. Here, $\dot{v}_{[i,k],n}$ represents the radial velocity of the $k$-th vehicle to the $i$-th RSU at time slot $n$. $\tilde{\nu}_{[i,k],n}$ and $\tilde{\mu}_{[i,k],n}$ are the measurements of $\nu_{[i,k],n}$ and $\mu_{[i,k],n}$, respectively. $\tilde{\mathbf{r}}_{[i,k],n}$ is the matched filtering output, which will be discussed in the next subsection. The vector $\mathbf{y}_{[i,k],n}$ can be obtained through matched filtering and further used to correct the prediction of $\mathbf{x}_{[i,k],n}$ and predict $\mathbf{x}_{[i,k],n+1}$. The block diagram of the DT-based ISAC network is shown in Fig. \ref{BlockDiagram}. To avoid extra overhead, the matched-filtering and Kalman filter are operated at the RSUs and only the states of the vehicle are needed to be transmitted to the processor for decision-making. The above processes can be operated either on a local DT or through edge computing, depending on the scale of the applications. In this work, we assume that the vehicle assignment and beamforming design are performed locally because there are stringent delay limits.

\begin{figure}[!t]
    \centering
    \includegraphics[width=0.7\columnwidth]{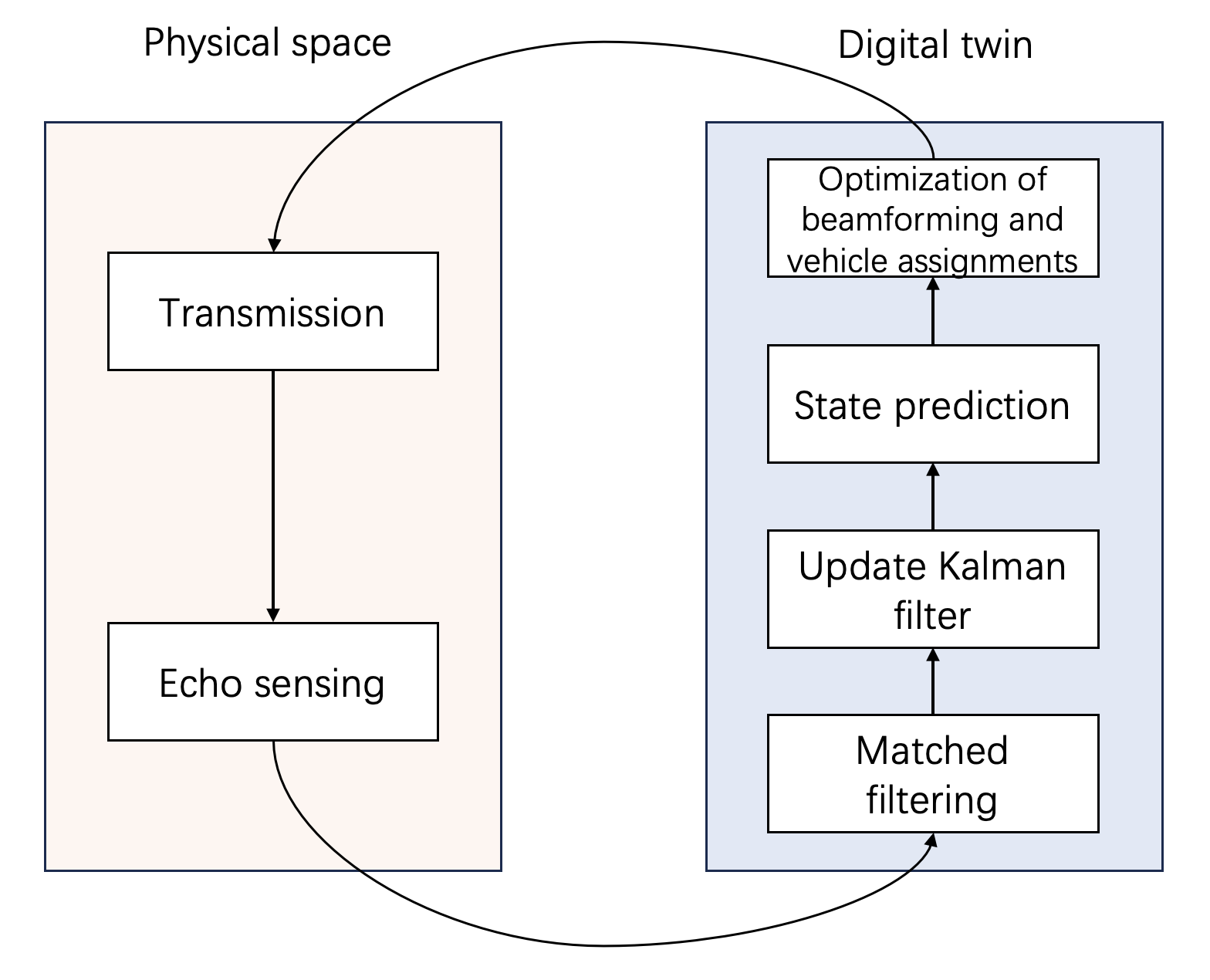}
    \caption{The block diagram of the DT-based beamforming design and vehicle assigning in an ISAC system.}
    \label{BlockDiagram}
\end{figure}

\subsection{Matched Filtering}

The sensing information of the vehicles can be extracted using matched filtering. For the rest of this subsection, we only consider the connected links between the vehicles and the RSUs. Using matched filtering method, $\nu_{[i,k],n}$ and $\mu_{[i,k],n}$ can be estimated by the following equation:
\begin{equation}
\begin{aligned}
&\{\tilde{\nu}_{[i,k],n}, \tilde{\mu}_{[i,k],n}\} = \\
&\arg \max_{\nu,\mu}\left|\int_{0}^{T_s} \mathbf{r}_{[i,k],n}(t) s^*_{[i,k],n}(t-\nu_{[i,k],n})e^{-j2\pi \mu_{[i,k],n} t}dt\right|^2,
\end{aligned}
\end{equation}
where $T_s$ is the duration of the ISAC signal, which should be lower than the slot duration $T$. The matched filter output of $\mathbf{r}_{[i,k],n}(t)$ can be written as
\begin{equation}
\begin{aligned}
    \tilde{\mathbf{r}}_{[i,k],n}=&\int_{0}^{T_s} \mathbf{r}_{[i,k],n}(t) s^*_{[i,k],n}(t-\tilde{\nu}_{[i,k],n})e^{-j2\pi \tilde{\mu}_{[i,k],n} t}dt\\
    =& \kappa \beta_{[i,k],n} \mathbf{b}(\varphi_{[i,k],n}) \mathbf{a}^{H}(\varphi_{[i,k],n})\mathbf{f}_{[i,k],n} \\
    &\cdot \int_{0}^{T_s}s_{[i,k],n}(t-\nu_{[i,k],n})s^*_{k,n}(t-\tilde{\nu}_{[i,k],n})\\
    &\cdot e^{-j2\pi(\tilde{\mu}_{[i,k],n}-\mu_{[i,k],n})t}dt\\
    &+\int_{0}^{T_s} \mathbf{z}_{[i,k],n}(t)s^*_{[i,k],n}(t-\tilde{\nu}_{[i,k],n})e^{-j2\pi\tilde{\mu}_{[i,k],n}t}dt\\
    =& \kappa \beta_{[i,k],n}G \mathbf{b}(\varphi_{[i,k],n})\mathbf{a}^{H}(\varphi_{[i,k],n})\mathbf{f}_{[i,k],n}+\tilde{\mathbf{z}}_{[i,k],n},
\end{aligned}
\end{equation}
where $\tilde{\mathbf{z}}_{[i,k],n}\sim \mathcal{CN} (0,\sigma_{\mathbf{r},[i,k],n}^2\mathbf{I}_{N_R})$ denotes the measurement noise after matched filtering with $\mathbf{I}_{N_R}$ representing a $N_R\times N_R$ identity matrix, and $G=\int_{0}^{T_s}s_{k,n}(t-\nu_{[i,k],n})s^*_{k,n}(t-\tilde{\nu}_{[i,k],n})e^{-j2\pi(\tilde{\mu}_{[i,k],n}-\mu_{[i,k],n})t}dt$ represents the matched filtering gain.


Note that $d_{[i,k],n}$ and $\dot{v}_{[i,k],n}$ can be reflected by $\tilde{\nu}_{[i,k],n}$ and $\tilde{\mu}_{[i,k],n}$, which yields:
\begin{equation}
    \tilde{\nu}_{[i,k],n} = \frac{2d_{[i,k],n}}{c}+\varepsilon_{[i,k],n},
\end{equation}
\begin{equation}
    \tilde{\mu}_{[i,k],n} = \frac{2\dot{v}_{[i,k],n}f_c}{c}+\varrho_{[i,k],n},
\end{equation}
where $f_c$ is the carrier frequency, $c$ is the speed of light, $\varepsilon_{[i,k],n}\sim\mathcal{N}(0,\sigma_{\nu,[i,k],n}^2)$ and $\varrho_{[i,k],n}\sim\mathcal{N}(0,\sigma_{\mu,[i,k],n}^2)$ are the estimation errors of $\tilde{\nu}_{[i,k],n}$ and $\tilde{\mu}_{[i,k],n}$, with noise variance being $\sigma_{\nu,[i,k],n}^2$ and $\sigma_{\mu,[i,k],n}^2$, respectively.

After matched filtering, the desired signal is amplified and the interference is filtered. The measurement noise can be seen as inversely proportional to the signal-to-noise ratio (SNR) at the receive antenna \cite{10.5555/151045}, which yields:
\begin{equation}
    \sigma_{\mathbf{r},[i,k],n}^2 = \frac{\rho_{\mathbf{r}}^2\sigma_e^2}{G},
\end{equation}
\begin{equation}
    \sigma_{\nu,[i,k],n}^2 = \frac{\rho_\nu^2\sigma_e^2}{G\kappa^2 |\beta_{[i,k],n}|^2|\eta_{[i,k],n}|^2},
\end{equation}
\begin{equation}
    \sigma_{\mu,[i,k],n}^2 = \frac{\rho_\mu^2\sigma_e^2}{G\kappa^2 |\beta_{[i,k],n}|^2|\eta_{[i,k],n}|^2},
\end{equation}
where $\eta_{[i,k],n}=\mathbf{a}^{H}(\varphi_{[i,k],n})\mathbf{f}_{[i,k],n}$ represent the beamforming gain factor, $\rho_{\mathbf{r}}$, $\rho_\nu$ and $\rho_\mu$ are constant indexes determined by the system configuration.

\subsection{State Evolution}

\begin{figure}[!t]
    \centering
    \includegraphics[width=0.9\columnwidth]{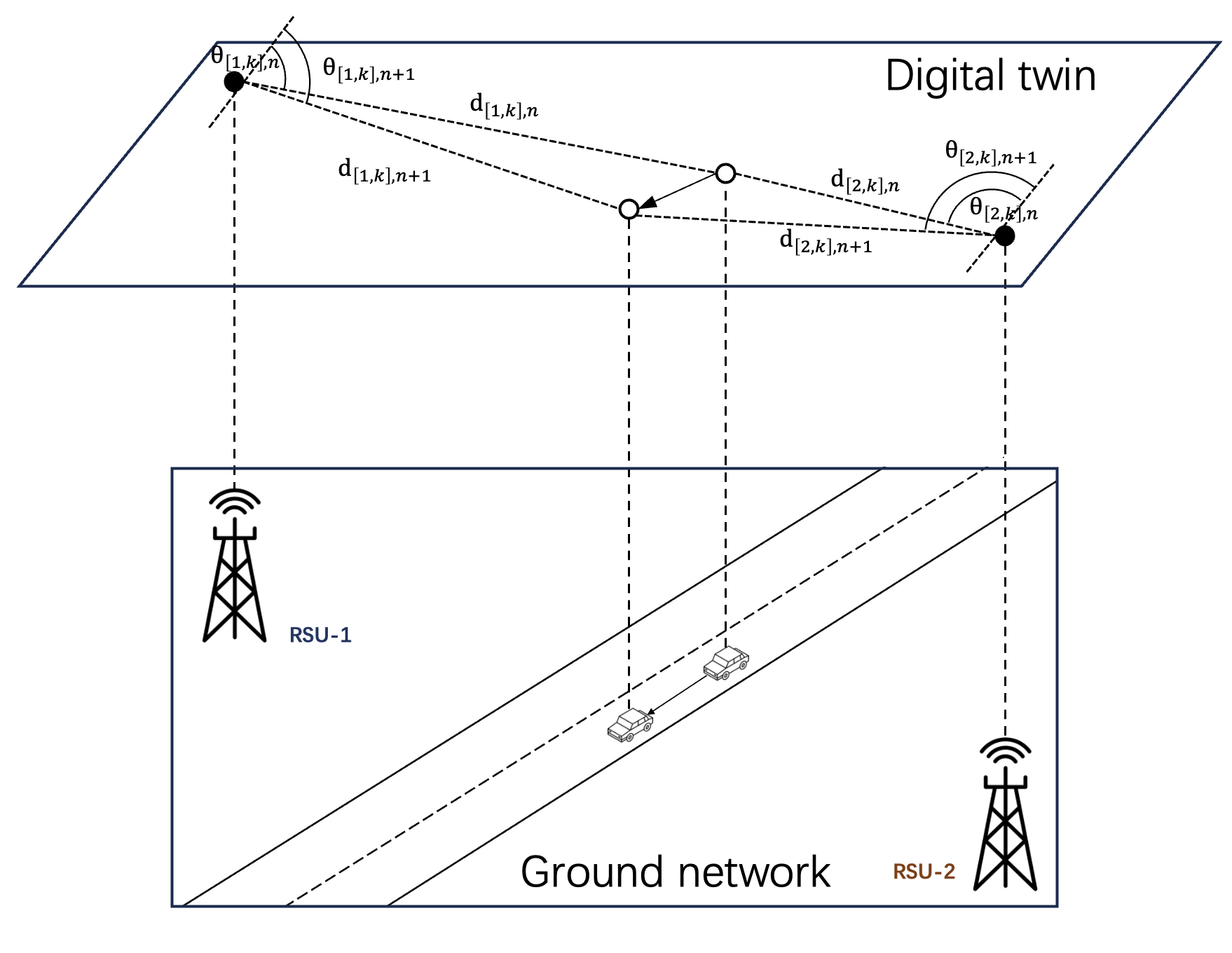}
    \caption{The kinematic model of a moving vehicle in the network.}
    \label{kinematics}
\end{figure}

At each time slot, the state estimation is based on the observation of radar echoes. The kinematic model of a moving vehicle in the network relative to the RSU is shown in Fig. \ref{kinematics}.
The prediction of $\mathbf{x}_{[i,k],n}$, denoted by $\hat{\mathbf{x}}_{[i,k],n}$, is obtained based on the measurement at time slot $n-1$, i.e., $\tilde{\mathbf{x}}_{[i,k],n-1}$. $\mathbf{x}_{[i,k],n}$ is a function of $\mathbf{x}_{[i,k],n-1}$, which yields:
\begin{equation}
    \mathbf{x}_{[i,k],n} = \mathbf{g}(\mathbf{x}_{[i,k],n-1})+\mathbf{\epsilon}_x,
\label{g}
\end{equation}
where function $\mathbf{g}(\cdot)$ can be written as (\ref{State Evo}), as shown at the top of the next page, and $\mathbf{\epsilon}_x=[\epsilon_{\varphi},\epsilon_{d},\epsilon_{\dot{v}}]^T$ is the prediction noise, with $\epsilon_{\varphi}\sim\mathcal{N}(0,\sigma_{\varphi}^2)$, $\epsilon_{d}\sim\mathcal{N}(0,\sigma_{d}^2)$, $\epsilon_{\dot{v}}\sim\mathcal{N}(0,\sigma_{\dot{v}}^2)$.

Since the vehicles are driven on a highway and the duration of each time slot is relatively short, it is reasonable to assume that the direction of the movement is along the highway and the velocities of the same vehicle at adjacent time slots are equal, i.e., $\dot{v}_{[i,k],n-1}/\varphi_{[i,k],n-1} = \dot{v}_{[i,k],n}/\varphi_{[i,k],n}$.

\newcounter{mytempeqncnt}
\begin{figure*}[t!]
\normalsize
\setcounter{mytempeqncnt}{\value{equation}}

\begin{equation}
\left\{
\begin{aligned}
& \varphi_{[i,k],n} = \frac{d_{[i,k],n-1}\dot{v}_{[i,k],n-1} T-(\dot{v}_{[i,k],n-1} T/\varphi_{[i,k],n-1})^2}{\dot{v}_{[i,k],n-1} T/\varphi_{[i,k],n-1} \times \sqrt{d_{[i,k],n-1}^2+(\dot{v}_{[i,k],n-1} T/\varphi_{[i,k],n-1})^2-2d_{[i,k],n-1}\dot{v}_{[i,k],n-1} T}}\\
& d_{[i,k],n} = \sqrt{d_{[i,k],n-1}^2+(\dot{v}_{[i,k],n-1} T/\varphi_{[i,k],n-1})^2-2d_{[i,k],n-1}\dot{v}_{[i,k],n-1} T}\\
& \dot{v}_{[i,k],n} = \frac{d_{[i,k],n-1}\dot{v}_{[i,k],n-1} T-(\dot{v}_{[i,k],n-1} T/\varphi_{[i,k],n-1})^2}{\dot{v}_{[i,k],n-1} T \times \sqrt{d_{[i,k],n-1}^2+(\dot{v}_{[i,k],n-1} T/\varphi_{[i,k],n-1})^2-2d_{[i,k],n-1}\dot{v}_{[i,k],n-1} T}}\times \dot{v}_{[i,k],n-1}
\end{aligned}
\right.
\label{State Evo}
\end{equation}
\hrulefill
\end{figure*}

\subsection{Extended Kalman Filter}

As the movement of the vehicles is continuous, the RSUs need to continuously track the vehicles for a certain period. Without sufficient information, sensing performance is limited. However, accuracy is expected to improve over time based on previous measurements. Therefore, this work applies the Kalman filter to track the vehicles. Since the observation function is non-linear, the standard Kalman filter is not applicable in this work. Instead, an EKF is utilized to estimate sensing errors. $\mathbf{G}_{[i,k],n}$ is defined as the Jacobian matrix of $\mathbf{g}(\tilde{\mathbf{x}}_{[i,k],n-1})$, which can be expressed as
\begin{equation}
\mathbf{G}_{[i,k],n} = 
\begin{bmatrix}
\frac{\partial \varphi_{[i,k],n}}{\partial \tilde{\varphi}_{[i,k],n-1}} & \frac{\partial \varphi_{[i,k],n}}{\partial \tilde{d}_{[i,k],n-1}} & \frac{\partial \varphi_{[i,k],n}}{\partial \tilde{\dot{v}}_{[i,k],n-1}} \\
\frac{\partial d_{[i,k],n}}{\partial \tilde{\varphi}_{[i,k],n-1}} & \frac{\partial d_{[i,k],n}}{\partial \tilde{d}_{[i,k],n-1}} & \frac{\partial d_{[i,k],n}}{\partial \tilde{\dot{v}}_{[i,k],n-1}} \\
\frac{\partial \dot{v}_{[i,k],n}}{\partial \tilde{\varphi}_{[i,k],n-1}}& \frac{\partial \dot{v}_{[i,k],n}}{\partial \tilde{d}_{[i,k],n-1}} &\frac{\partial \dot{v}_{[i,k],n}}{\partial \tilde{\dot{v}}_{[i,k],n-1}}
\end{bmatrix}.
\end{equation}

Since $\mathbf{y}_{[i,k],n}$ is a function of $\mathbf{x}_{[i,k],n-1}$, we have:
\begin{equation}
    \mathbf{y}_{[i,k],n} = \mathbf{h}(\mathbf{x}_{[i,k],n})+\mathbf{\chi}_{[i,k],n},
\end{equation}
where $\mathbf{\chi}_{[i,k],n}\in \mathbb{C}^{(N_t+2)\times 1}$ is the measurement noise. The Jacobian matrix of $\mathbf{h}(\mathbf{x}_{[i,k],n})$, i.e., $\mathbf{H}_{[i,k],n}$, can be expressed as
\begin{equation}
\begin{aligned}
\mathbf{H}_{[i,k],n}&=\frac{\partial \mathbf{h}(\mathbf{x}_{[i,k],n})}{\partial \mathbf{x}_{[i,k],n}}\\
&=
\begin{bmatrix}
\frac{\partial \tilde{\mathbf{r}}_{[i,k],n}}{\partial \varphi_{[i,k],n}} & 0 & 0 \\
0 & \frac{2}{c} & 0 \\
0 & 0 & \frac{2f_c}{c} 
\end{bmatrix}.
\end{aligned}
\end{equation}

By taking the first-order Taylor expansion, the state transition function and the observation function can be approximately linearized by
\begin{equation}
\begin{aligned}
    \mathbf{x}_{[i,k],n} \approx & \:\mathbf{G}_{[i,k],n} \mathbf{x}_{[i,k],n-1} + \mathbf{g}(\tilde{\mathbf{x}}_{[i,k],n-1})\\
    &-\mathbf{G}_{[i,k],n-1} \tilde{\mathbf{x}}_{[i,k],n-1}+\epsilon_{x},
\end{aligned}
\end{equation}
\begin{equation}
\begin{aligned}
    \mathbf{y}_{[i,k],n}\approx & \: \mathbf{H}_{[i,k],n} \mathbf{x}_{[i,k],n} + \mathbf{h}(\hat{\mathbf{x}}_{[i,k],n})\\
    &- \mathbf{H}_{[i,k],n} \hat{\mathbf{x}}_{[i,k],n}+ \mathbf{\chi}_{[i,k],n}.
\end{aligned}
\end{equation}

The state prediction MSE matrix $\hat{\mathbf{M}}_{[i,k],n}$ can be expressed as
\begin{equation}
\begin{aligned}
    \hat{\mathbf{M}}_{[i,k],n}  =& \mathbb{E}[(\mathbf{x}_{[i,k],n}-\hat{\mathbf{x}}_{[i,k],n})(\mathbf{x}_{[i,k],n}-\hat{\mathbf{x}}_{[i,k],n})^H]\\
    =&\mathbf{G}_{[i,k],n} \tilde{\mathbf{M}}_{[i,k],n-1}\mathbf{G}_{[i,k],n}^H + \mathbf{E}_{[i,k],n},
\end{aligned}
\label{M1}
\end{equation}
where $\mathbf{E}_{[i,k],n}=\text{diag}(\sigma_{\varphi}^2,\sigma_{d}^2,\sigma_{\dot{v}}^2)$ is the covariance matrix of the state prediction noise, and $\tilde{\mathbf{M}}_{[i,k],n-1}$ is the state measurement MSE matrix at the $(n-1)$-th time slot.

The Kalman gain $\mathbf{K}_{[i,k],n}$ can be expressed as
\begin{equation}
\begin{aligned}
    \mathbf{K}_{[i,k],n} = &\hat{\mathbf{M}}_{[i,k],n} \mathbf{H}_{[i,k],n}^H \\
    &(\mathbf{Q}_{[i,k],n}+\mathbf{H}_{[i,k],n} \hat{\mathbf{M}}_{[i,k],n}\mathbf{H}_{[i,k],n}^H)^{-1}.
\end{aligned}
\label{KalmanGain}
\end{equation}

We define $\mathbf{e}_{[i,k],n} = \mathbf{y}_{[i,k],n}-\mathbf{h}(\hat{\mathbf{x}}_{[i,k],n})$ as the difference between the actual received echo signal and the anticipated echo signal. To enhance the accuracy of the vehicle tracking, $\mathbf{e}_{[i,k],n}$ is measured after each time slot and subsequently updated in the Kalman gain. Through analyzing the error of the state prediction, we can obtain the following lemma.

\begin{lemma}
The error of the state prediction, i.e., $\mathbf{x}_{[i,k],n} - \hat{\mathbf{x}}_{[i,k],n}$ equals to $\mathbf{H}_{L,[i,k],n}^{-1}(\mathbf{e}_{[i,k],n}-\mathbf{\chi}_{[i,k],n})$, where $\mathbf{H}_{\text{L},[i,k],n}^{-1}=(\mathbf{H}_{[i,k],n}^H\mathbf{H}_{[i,k],n})^{-1}\mathbf{H}_{[i,k],n}^H$ is the left inverse of $\mathbf{H}_{[i,k],n}$.
\end{lemma}
\begin{proof}
$\mathbf{e}_{[i,k],n}$ denotes the difference between $\mathbf{y}_{[i,k],n}$ and $\mathbf{h}(\hat{\mathbf{x}}_{[i,k],n})$ and can be expressed as
\begin{equation}
\begin{aligned}
    \mathbf{e}_{[i,k],n} = &\mathbf{H}_{[i,k],n} \mathbf{x}_{[i,k],n} - \mathbf{H}_{[i,k],n}\hat{\mathbf{x}}_{[i,k],n}  + \mathbf{\chi}_{[i,k],n} \\
    =& \mathbf{H}_{[i,k],n} (\mathbf{x}_{[i,k],n}-\hat{\mathbf{x}}_{[i,k],n}) + \mathbf{\chi}_{[i,k],n}.
\end{aligned}
\end{equation}
Then, we have:
\begin{equation}
\mathbf{x}_{[i,k],n}-\hat{\mathbf{x}}_{[i,k],n} = \mathbf{H}_{L,[i,k],n}^{-1}(\mathbf{e}_{[i,k],n}-\mathbf{\chi}_{[i,k],n}).
\end{equation}

Hence, Lemma 1 is proved.
\end{proof}

The updated error matrix can be expressed as

\begin{equation}
\begin{aligned}
&\mathbb{E}[(\mathbf{x}_{[i,k],n}-\hat{\mathbf{x}}_{[i,k],n})(\mathbf{x}_{[i,k],n}-\hat{\mathbf{x}}_{[i,k],n})^H]\\
& = \mathbb{E}\Big[\mathbf{H}_{L,[i,k],n}^{-1}(\mathbf{e}_{[i,k],n}-\mathbf{\chi}_{[i,k],n}) \\
& \:\:\:\:\:\:(\mathbf{e}_{[i,k],n}-\mathbf{\chi}_{[i,k],n})^H(\mathbf{H}_{L,[i,k],n}^{-1})^H\Big]\\
& = \mathbf{H}_{\text{L},[i,k],n}^{-1} \Big(\mathbf{e}_{[i,k],n}\mathbf{e}_{[i,k],n}^H - \mathbb{E}[\mathbf{e}_{[i,k],n}\mathbf{\chi}_{[i,k],n}^H]\\
&\:\:\:\:\: - \mathbb{E}[\mathbf{\chi}_{[i,k],n}\mathbf{e}_{[i,k],n}^H]+ \mathbb{E}[\mathbf{\chi}_{[i,k],n}\mathbf{\chi}_{[i,k],n}^H] \Big) (\mathbf{H}_{\text{L},[i,k],n}^{-1})^{H} \\
& = \mathbf{H}_{\text{L},[i,k],n}^{-1}(\mathbf{e}_{[i,k],n}\mathbf{e}_{[i,k],n}^H-\mathbf{Q}_{[i,k],n})(\mathbf{H}_{\text{L},[i,k],n}^{-1})^H,
\end{aligned}
\label{M_pre}
\end{equation}
where $\mathbf{Q}_{[i,k],n}=\text{diag}(\sigma_{\mathbf{r},[i,k],n}^2\mathbf{1}_{N_r}^T,\sigma_{\nu,[i,k],n}^2,\sigma_{\mu,[i,k],n}^2)$ with $\mathbf{1}_{N_r}$ representing a size-$N_r$ all one column vector.

Substituting (\ref{M_pre}) into (\ref{KalmanGain}), the updated Kalman gain can be written as
\begin{equation}
\begin{aligned}
    &\mathbf{K}_{[i,k],n} = \mathbf{H}_{\text{L},[i,k],n}^{-1}(\mathbf{e}_{[i,k],n}\mathbf{e}_{[i,k],n}^H-\mathbf{Q}_{[i,k],n})\\
    &(\mathbf{H}_{[i,k],n}\mathbf{H}_{\text{L},[i,k],n}^{-1})^H\Big(\mathbf{Q}_{[i,k],n}+\mathbf{H}_{[i,k],n}\mathbf{H}_{\text{L},[i,k],n}^{-1}\\
    &(\mathbf{e}_{[i,k],n}\mathbf{e}_{[i,k],n}^H-\mathbf{Q}_{[i,k],n})(\mathbf{H}_{[i,k],n}\mathbf{H}_{\text{L},[i,k],n}^{-1})^H \Big)^{-1}.
\end{aligned}
\label{Kalman gain}
\end{equation}

This updated Kalman gain can be used for fix some of the prediction errors at the last time slot, which yields:
\begin{equation}
    \tilde{\mathbf{x}}_{[i,k],n} = \hat{\mathbf{x}}_{[i,k],n} + \mathbf{K}_{[i,k],n} \mathbf{e}_{[i,k],n}.
\end{equation}

Finally, the MSE matrix of the state measurement can be computed by
\begin{equation}
\begin{aligned}
    \tilde{\mathbf{M}}_{[i,k],n} =& \mathbb{E} [(\mathbf{x}_{[i,k],n} - \tilde{\mathbf{x}}_{[i,k],n})(\mathbf{x}_{[i,k],n} - \tilde{\mathbf{x}}_{[i,k],n})^H]\\
    =& \mathbb{E} [(\mathbf{x}_{[i,k],n} -\hat{\mathbf{x}}_{[i,k],n} - \mathbf{K}_{[i,k],n} \mathbf{e}_{[i,k],n})\\
    &\times (\mathbf{x}_{[i,k],n} - \hat{\mathbf{x}}_{[i,k],n} - \mathbf{K}_{[i,k],n} \mathbf{e}_{[i,k],n})^H]\\
    =& \mathbb{E}[(\mathbf{x}_{[i,k],n}-\hat{\mathbf{x}}_{[i,k],n})(\mathbf{x}_{[i,k],n}-\hat{\mathbf{x}}_{[i,k],n})^H] \\
    &- \mathbb{E}[\mathbf{K}_{[i,k],n}\mathbf{e}_{[i,k],n}  (\mathbf{x}_{[i,k],n} -\hat{\mathbf{x}}_{[i,k],n})^H]\\
    =& \mathbf{K}_{[i,k],n} \mathbf{Q}_{[i,k],n}(\mathbf{H}_{\text{L},[i,k],n}^{-1})^H.
\end{aligned}
\label{measurement}
\end{equation}

\begin{algorithm}[t!]
\caption{EKF-based State Prediction}
\begin{algorithmic}[1]
\State Input $\hat{\mathbf{x}}_{[i,k],n}$, $\mathbf{y}_{[i,k],n}$, $\mathbf{H}_{[i,k],n}$, $\mathbf{Q}_{[i,k],n}$.
\State Compare the difference between $\mathbf{y}_{[i,k],n}$ and $\mathbf{h}(\hat{\mathbf{x}}_{[i,k],n})$, save it as $\mathbf{e}_{[i,k],n}$.
\State Compute the Kalman gain $\mathbf{K}_{[i,k],n}$ by (\ref{Kalman gain}).
\State Based on the Kalman gain and received echo signal, correct the predicted state to achieve the measured state $\tilde{\mathbf{x}}_{[i,k],n}$ by (\ref{measurement}).
\State Predict the vehicle state at time slot $n+1$ by $\hat{\mathbf{x}}_{[i,k],n+1} = \mathbf{g}(\tilde{\mathbf{x}}_{[i,k],n})$.
\State Return $\hat{\mathbf{x}}_{[i,k],n+1}$.
\end{algorithmic}
\label{EKF}
\end{algorithm}

The DT then uses the predicted states and the MSE matrix to design the predictive beamforming and assign vehicles to different RSUs. The overall algorithm for state prediction with EKF is summarized in Algorithm \ref{EKF}.

\section{Optimization Problem for Beamforming Design and Vehicle Assignment}

Our goal is to optimize the overall transmission rate within the vehicular network by strategically assigning vehicles to different RSUs and designing predictive beamforming for both sensing and communication purposes. Once the DT predicts the location of vehicles at the next time slot, we start the beamforming design and vehicle allocation process. Ideally, we would aim to find the optimal $\mathbf{F}_{i,n}$ and $\mathbf{\xi}_n=\{\xi_{[i,k],n}\}, \forall i, \forall k$ that maximize both the communication rate and the sensing accuracy. However, there exists a clear trade-off between these objectives. Greedily designing the predictive beamforming that is highly suitable for communication results in a maximum communication rate in the next time slot. However, this design could lead to a decline in the CRB of the next time slot which will unavoidably affect the accuracy of the DT for the subsequent time slot, subsequently affecting the potential communication rate in later time slots. The challenge of establishing a direct relationship between sensing accuracy and communication rate makes the problem more complex to model.

\subsection{Posterior Cramér-Rao Lower Bound}

CRB is a fundamental concept in estimation theory and statistics. It establishes the minimum variance level for unbiased estimators of a deterministic parameter that is fixed but unknown. Specifically, the variance of any such estimator cannot be lower than the reciprocal of the Fisher information. However, in this work, the estimation errors of vehicle states depend on not only the measured parameters but also the errors inherited from the previous time slots. Therefore, the PCRB is used for finding the MSE lower bound of the vehicle state.

The conditional probability density function (PDF) of $\mathbf{y}_{[i,k],n}$ and $\mathbf{x}_{[i,k],n}$ given $\hat{\mathbf{x}}_{[i,k],n}$ can be expressed as
\begin{equation}
\begin{aligned}
p(\mathbf{x}_{[i,k],n},\mathbf{y}_{[i,k],n}&|\hat{\mathbf{x}}_{[i,k],n}) = \\
&p(\mathbf{y}_{[i,k],n}|\mathbf{x}_{[i,k],n},\hat{\mathbf{x}}_{[i,k],n})p(\mathbf{x}_{[i,k],n}|\hat{\mathbf{x}}_{[i,k],n}),
\end{aligned}
\end{equation}
where $p(\mathbf{y}_{[i,k],n}|\mathbf{x}_{[i,k],n},\hat{\mathbf{x}}_{[i,k],n})$ is the conditional PDF of $\mathbf{y}_{[i,k],n}$ given $\mathbf{x}_{[i,k],n}$ and $\hat{\mathbf{x}}_{[i,k],n}$, and $p(\mathbf{x}_{[i,k],n}|\hat{\mathbf{x}}_{[i,k],n})$ is the conditional PDF of $\mathbf{x}_{[i,k],n}$ given $\hat{\mathbf{x}}_{[i,k],n}$. According to the measurement model, $p(\mathbf{y}_{[i,k],n}|\mathbf{x}_{[i,k],n},\hat{\mathbf{x}}_{[i,k],n})$ can be computed by
\begin{equation}
\begin{aligned}
p(\mathbf{y}_{[i,k],n}|\mathbf{x}_{[i,k],n}&,\hat{\mathbf{x}}_{[i,k],n}) = \frac{1}{\pi ^{N_r+2}\det (\mathbf{Q}_{[i,k],n})}\\
&\exp \Big(\big(\mathbf{y}_{[i,k],n} - \mathbf{h}(\mathbf{x}_{[i,k],n})\big)^H\mathbf{Q}_{[i,k],n}^{-1}\\
&(\mathbf{y}_{[i,k],n} - \mathbf{h}(\mathbf{x}_{[i,k],n}))\Big).
\end{aligned}
\end{equation}

The posterior Fisher information matrix $\mathbf{F}_{[i,k],n}$ can be computed by
\begin{equation}
\begin{aligned}
    &\mathbf{F}_{[i,k],n}\\
    &=-\mathbb{E}\left[\frac{\partial^2\ln p(\mathbf{y}_{[i,k],n}|\mathbf{x}_{[i,k],n},\hat{\mathbf{x}}_{[i,k],n})p(\mathbf{x}_{[i,k],n}|\hat{\mathbf{x}}_{[i,k],n})}{\partial \mathbf{x}_{[i,k],n}^2}  \right]\\
    &= -\mathbb{E}\left[\frac{\partial^2\ln p(\mathbf{y}_{[i,k],n}|\mathbf{x}_{[i,k],n},\hat{\mathbf{x}}_{[i,k],n})}{\partial \hat{\mathbf{x}}_{[i,k],n}^2}  \right]\\
    &\:\:\:\:\:\: -\mathbb{E}\left[\frac{\partial^2\ln p(\mathbf{x}_{[i,k],n}|\hat{\mathbf{x}}_{[i,k],n})}{\partial \mathbf{x}_{[i,k],n}^2}   \right]\\
    &= \mathbf{H}_{[i,k],n}\mathbf{Q}_{[i,k],n}^{-1}\mathbf{H}_{[i,k],n}^H+ \hat{\mathbf{M}}_{[i,k],n}^{-1}.
\end{aligned}
\end{equation}

The PCRB matrix is equivalent to the state measurement MSE matrix, hence we have:
\begin{equation}
\begin{aligned}
    &\tilde{\mathbf{M}}_{[i,k],n} = \\
    &(\mathbf{H}_{[i,k],n}\mathbf{Q}_{[i,k],n}^{-1}\mathbf{H}_{[i,k],n}^H+ (\mathbf{G}_{[i,k],n} \tilde{\mathbf{M}}_{[i,k],n-1}\mathbf{G}_{[i,k],n}^H)^{-1})^{-1},
\end{aligned}
\label{inverse}
\end{equation}
which shows the direct relationship between $\tilde{\mathbf{M}}_{[i,k],n}$ and $\tilde{\mathbf{M}}_{[i,k],n-1}$.

\subsection{Problem Formulation}

While our primary focus is on communication, ensuring a high level of sensing accuracy is crucial for constructing the DT, and, consequently, leads to reliable communication. In this work, although our goal is to maximize the overall throughput at the current time slot, we also impose a constraint on the PCRB to ensure that the sensing error does not increase over time, thereby enabling the construction of an accurate DT. Therefore, we set $\tilde{m}_{[i,k],n}^{(11)} \leq \tilde{m}_{[i,k],n-1}^{(11)}$, where $m_{[i,k],n}^{(ij)}$ denotes the $(i,j)$-th entry of $\tilde{\mathbf{M}}_{[i,k],n}$.

$\tilde{\mathbf{M}}_{[i,k],n}$ is a $3\times 3$ matrix, hence we can easily find its inverse matrix. Since we assume that the current sensing performance is good, $\eta_{[i,k],n} \approx 1$. With all the other entries being constant, we have:
\begin{equation}
    \left(\mathbf{\Pi}_{[i,k],n}\mathbf{f}_{[i,k],n}\right)^H \left(\mathbf{\Pi}_{[i,k],n}\mathbf{f}_{[i,k],n}\right) \geq \Lambda_{[i,k],n},
\label{Constraint}
\end{equation}
where $\Lambda_{[i,k],n}$ is a constant value achieved by using (\ref{inverse}) to find the inverse of the $3\times3$ matrix, and $\mathbf{\Pi}_{[i,k],n} = \partial \mathbf{b}(\varphi_{[i,k],n})\mathbf{a}^{H}(\varphi_{[i,k],n})/\partial \varphi_{[i,k],n}$. The calculation of $\Lambda_{[i,k],n}$ is detailed in the Appendix. Then, we can formulate the optimization problem as (\ref{Optimization}), shown at the top of the next page.

\begin{figure*}[t!]
\normalsize
\setcounter{mytempeqncnt}{\value{equation}}

\begin{equation}    
\begin{aligned}
    \max_{\mathbf{\xi}_n,\mathbf{F}_n}&\sum_{i=1}^{I}\sum_{k=1}^{K}\log\left(1+\frac{\kappa'^2|\alpha_{[i,k],n}||\mathbf{a}^{H}(\varphi_{[i,k],n})\mathbf{f}_{[i,k],n}|^2\xi_{[i,k],n}}{\sum_{m\neq k}^{K}\kappa'^2|\alpha_{[i,k],n}||\mathbf{a}^{H}(\varphi_{[i,k],n})\mathbf{f}_{[i,m],n}|^2\xi_{[i,m],n}+\sigma_c^2}\right)\\
    \text{subject to}\:\: &  \xi_{[i,k],n}\in \{0,1\},\forall k,\forall i\\
    &\sum_{i=1}^{I}\xi_{[i,k],n}=1, \forall k\\
    &||\mathbf{f}_{[i,k],n}||^2\leq1, \forall k, \forall i\\
    & \Lambda_{[i,k],n}-\left(\frac{\partial \mathbf{b}(\varphi_{[i,k],n})\mathbf{a}^{H}(\varphi_{[i,k],n})}{\partial \varphi_{[i,k],n}}\mathbf{f}_{[i,k],n}\right)^H \left(\frac{\partial \mathbf{b}(\varphi_{[i,k],n})\mathbf{a}^{H}(\varphi_{[i,k],n})}{\partial \varphi_{[i,k],n}}\mathbf{f}_{[i,k],n}\right) \leq 0, \forall k, \forall i
\end{aligned}
\label{Optimization}
\end{equation}
\end{figure*}

\section{Algorithm Design}

Problem (\ref{Optimization}) is difficult to solve because it is non-convex. To solve Problem (\ref{Optimization}), we use an auxiliary variable $\gamma_{[i,k],n}$ to replace $\text{SINR}_{[i,k],n}$. This allows us to formulate a new optimization problem as Problem (\ref{Optimization2}), shown at the top of the next page below Problem (\ref{Optimization}).

\begin{figure*}[t!]
\begin{equation}
\begin{aligned}
    \max_{\mathbf{\xi}_n,\mathbf{F}_n,\mathbf{\gamma}_n} & \sum_{i=1}^{I} \sum_{k=1}^{K} \log (1+\gamma_{[i,k],n})\\
    \text{subject to}\:\: &  \xi_{[i,k],n}\in \{0,1\},\forall k,\forall i\\
    &\sum_{i=1}^{I}\xi_{[i,k],n}=1, \forall k\\
    &||\mathbf{f}_{[i,k],n}||^2\leq1,\forall k,\forall i\\
    & \Lambda_{[i,k],n}-\left(\frac{\partial \mathbf{b}(\varphi_{[i,k],n})\mathbf{a}^{H}(\varphi_{[i,k],n})}{\partial \varphi_{[i,k],n}}\mathbf{f}_{[i,k],n}\right)^H \left(\frac{\partial \mathbf{b}(\varphi_{[i,k],n})\mathbf{a}^{H}(\varphi_{[i,k],n})}{\partial \varphi_{[i,k],n}}\mathbf{f}_{[i,k],n}\right) \leq 0,\forall k,\forall i\\
    &\gamma_{[i,k],n} \leq \frac{\kappa'^2|\alpha_{[i,k],n}||\mathbf{a}^{H}(\varphi_{[i,k],n})\mathbf{f}_{[i,k],n}|^2\xi_{[i,k],n}}{\sum_{m\neq k}^{K}\kappa'^2|\alpha_{[i,k],n}||\mathbf{a}^{H}(\varphi_{[i,k],n})\mathbf{f}_{[i,m],n}|^2\xi_{[i,m],n}+\sigma_c^2},\forall k,\forall i\\
\end{aligned}
\label{Optimization2}
\end{equation}
\end{figure*}

Problem (\ref{Optimization2}) can be divided into an inner optimization problem over $\mathbf{\gamma}_n$, and an outer optimization problem over $\mathbf{\xi}_n$ and $\mathbf{F}_n$. The inner optimization problem is convex over $\mathbf{\gamma}_n$, hence, strong duality holds, and we can formulate the Lagrangian function as (\ref{Lagrangian}), shown at the top of this page.

\begin{figure*}[t!]
\begin{equation}
\begin{aligned}
    L(\mathbf{\gamma}_n,\mathbf{\lambda}_n)& =  \sum_{i=1}^{I} \sum_{k=1}^{K} \log (1+\gamma_{[i,k],n})\\
    &-\sum_{i=1}^{I} \sum_{k=1}^{K} \lambda_{[i,k],n} \left(\gamma_{[i,k],n}-\frac{\kappa'^2|\alpha_{[i,k],n}||\mathbf{a}^{H}(\varphi_{[i,k],n})\mathbf{f}_{[i,k],n}|^2\xi_{[i,k],n}}{\sum_{m\neq k}^{K}\kappa'^2|\alpha_{[i,k],n}||\mathbf{a}^{H}(\varphi_{[i,k],n})\mathbf{f}_{[i,m],n}|^2\xi_{[i,m],n}+\sigma_c^2} \right)\\
\end{aligned}
\label{Lagrangian}
\end{equation}
\hrulefill
\end{figure*}

Problem (\ref{Optimization2}) is equivalent to the following dual problem:
\begin{equation}
    \min_{\mathbf{\lambda}_n\geq0} \max_{\mathbf{\gamma}_n} L(\mathbf{\gamma}_n,\mathbf{\lambda}_n).
\end{equation}

The saddle point $(\mathbf{\gamma}_n^*,\mathbf{\lambda}_n^*)$ can be found by setting $\partial L(\mathbf{\gamma}_n,\mathbf{\lambda}_n)/\partial \gamma_{[i,k],n}=0$, which yields:

\begin{equation}
\begin{aligned}
    \lambda_{[i,k],n}^* &= \frac{1}{1+\gamma_{[i,k],n}} \\
    =& \frac{\sum_{m\neq k}^{K}\kappa'^2|\alpha_{[i,k],n}||\mathbf{a}^{H}(\varphi_{[i,k],n})\mathbf{f}_{[i,m],n}|^2\xi_{[i,m],n}+\sigma_c^2}{\sum_{m=1}^{K}\kappa'^2|\alpha_{[i,k],n}||\mathbf{a}^{H}(\varphi_{[i,k],n})\mathbf{f}_{[i,m],n}|^2\xi_{[i,m],n}+\sigma_c^2}.
\end{aligned}
\label{AfterL}
\end{equation}

Since (\ref{AfterL}) is convex with respect to $\mathbf{\gamma}_n$, while the other parameters are fixed, the optimal $\mathbf{\gamma}_n$, i.e., $\mathbf{\gamma}_n^*$ can be calculated by

\begin{equation}
\begin{aligned}
    \gamma_{[i,k],n}^* &=\\
    &\frac{\kappa'^2|\alpha_{[i,k],n}||\mathbf{a}^{H}(\varphi_{[i,k],n})\mathbf{f}_{[i,k],n}|^2\xi_{[i,k],n}}{\sum_{m\neq1}^K \kappa'^2|\alpha_{[i,k],n}||\mathbf{a}^{H}(\varphi_{[i,k],n})\mathbf{f}_{[i,m],n}|^2\xi_{[i,k],n}+\sigma_c^2}.
    \label{gamma}
\end{aligned}
\end{equation}

Once $\mathbf{\gamma}_n$ is fixed, we use fractional programming \cite{Shen} to eliminate the fractional term in the objective function. We introduce a set of auxiliary variables $y_{[i,k],n}$, which enables us to express the transformed objective function as a new function denoted by $f_q(\mathbf{\xi}_n,\mathbf{F}_n,\mathbf{\gamma}_n,\mathbf{Y}_n)$, as shown by (\ref{fq}) at the top of the next page, where $\mathbf{Y}_n$ is the set of $\{y_{[i,k],n}\}$. Problem (\ref{Optimization2}) is equivalent to maximizing $f_q(\mathbf{\xi}_n,\mathbf{F}_n,\mathbf{\gamma}_n^*,\mathbf{Y}_n)$, hence we can formulate an optimization problem on $f_q(\mathbf{\xi}_n,\mathbf{F}_n,\mathbf{\gamma}_n^*,\mathbf{Y}_n)$, as shown by (\ref{Opt3}) at the top of the next page below (\ref{fq}).

\begin{figure*}[t!]
\begin{equation}
\begin{aligned}
    f_q(\mathbf{\xi}_n,\mathbf{F}_n,\mathbf{\gamma}_n,\mathbf{Y}_n)& = \sum_{i=1}^I \sum_{k=1}^K \Bigg( \log (1+\gamma_{[i,k],n})-\gamma_{[i,k],n}\\
    &+2y_{[i,k],n}\sqrt{1+\gamma_{[i,k],n}}\kappa'\sqrt{\alpha_{[i,k],n}}\mathbf{a}^{H}(\varphi_{[i,k],n})\mathbf{f}_{[i,k],n}\xi_{[i,k],n}\\
    &-\sum_{m=1}^Ky_{[i,k],n}^2\Big(\kappa'^2|\alpha_{[i,k],n}||\mathbf{a}^{H}(\varphi_{[i,k],n})\mathbf{f}_{[i,m],n}|^2\xi_{[i,m],n}+\sigma_c^2\Big)\Bigg)\\
\end{aligned}
\label{fq}
\end{equation}
\end{figure*}

\begin{figure*}[t!]
\begin{equation}
\begin{aligned}
    \max_{\mathbf{\xi}_{n},\mathbf{F}_n,\mathbf{Y}_n} & f_q(\mathbf{\xi}_n,\mathbf{F}_n,\mathbf{\gamma}_n^*,\mathbf{Y}_n) \\
    \text{subject to}\:\: &  \xi_{[i,k],n}\in \{0,1\} ,\forall k,\forall i\\
    &\sum_{i=1}^{I}\xi_{[i,k],n}=1 ,\forall k\\
    &\mathbf{f}_{[i,k],n}^H \mathbf{f}_{[i,k],n} \leq1 ,\forall k,\forall i\\
    & \Lambda_{[i,k],n}-\left(\frac{\partial \mathbf{b}(\varphi_{[i,k],n})\mathbf{a}^{H}(\varphi_{[i,k],n})}{\partial \varphi_{[i,k],n}}\mathbf{f}_{[i,k],n}\right)^H \left(\frac{\partial \mathbf{b}(\varphi_{[i,k],n})\mathbf{a}^{H}(\varphi_{[i,k],n})}{\partial \varphi_{[i,k],n}}\mathbf{f}_{[i,k],n}\right) \leq 0,\forall k,\forall i
\end{aligned}
\label{Opt3}
\end{equation}
\end{figure*}

To solve Problem (\ref{Opt3}), we use a heuristic algorithm that iteratively fixes the other parameters and optimizes one parameter at a time. As $f_q(\mathbf{\xi}_n,\mathbf{F}_n,\mathbf{\gamma}_n^*,\mathbf{Y}_n)$ is a convex function of $y_{[i,k],n}$, the optimal $y_{[i,k],n}$, i.e., $y_{[i,k],n}^*$ can be achieved by setting $\partial f_q(\mathbf{\xi}_n,\mathbf{F}_n,\mathbf{\gamma}_n,\mathbf{Y}_n)/\partial y_{[i,k],n}=0$. Then we have:
\begin{equation}
\begin{aligned}
    y_{[i,k],n}^*& =\\
    &\frac{\kappa|\beta_{[i,k],n}||\mathbf{a}^{H}(\varphi_{[i,k],n})\mathbf{f}_{[i,k],n}|\xi_{[i,k],n}}{\sum_{m\neq k}^K\kappa^2|\beta_{[i,m],n}|^2|\mathbf{a}^{H}(\varphi_{[i,k],n})\mathbf{f}_{[i,m],n}|^2\xi_{[i,m],n}+\sigma_z^2}.
    \label{x}
\end{aligned}
\end{equation}

\begin{figure*}[t!]

\begin{equation}
\begin{aligned}
    \zeta_{[i,k],n} =& \log(1+\gamma_{[i,k],n})-\gamma_{[i,k],n}+2y_{[i,k],n}\sqrt{1+\gamma_{[i,k],n}}\kappa'\sqrt{\alpha_{[i,k],n}}\mathbf{a}^{H}(\varphi_{[i,k],n})\mathbf{f}_{[i,k],n}\xi_{[i,k],n}\\
    &-\sum_{m=1}^Ky_{[i,k],n}^2\Big(\kappa'^2|\alpha_{[i,k],n}||\mathbf{a}^{H}(\varphi_{[i,k],n})\mathbf{f}_{[i,m],n}|^2\xi_{[i,m],n}+\sigma_c^2\Big)
\end{aligned}
\label{zeta}
\end{equation}
\hrulefill
\end{figure*}

With $y_{[i,k],n}$ being fixed too, the objective function and all the constraints of Problem (\ref{Opt3}) are convex functions with respect to $\mathbf{F}_n$. Hence, we can use the projected gradient descent (PGD) algorithm to find the solution.

The optimization of vehicle assignments is a binary assignment problem. We present two methods to optimize the matching between RSUs and vehicles. Firstly, a greedy algorithm is proposed that sequentially allocates vehicles for optimal performance. Secondly, a heuristic algorithm is proposed for assigning the vehicles to the RSUs. We define an index $\zeta_{[i,k],n}$ as (\ref{zeta}), as shown in the next page below (\ref{Opt3}), representing the comprehensive utility and penalty arising from the allocation of vehicle $k$ to RSU $i$ during time slot $n$.

\begin{algorithm}[t!]
\caption{Greedy Algorithm}
\begin{algorithmic}[1]
\State Define $\mathcal{V}_{un}$ as the set of vehicles that are yet unallocated in the system. Initialize set $\mathcal{V}=\emptyset$. 
\State Sort the vehicles in set $\mathcal{V}_{un}$ based on $|1/d_{[1,k],n}-1/d_{[2,k],n}|^2$.
\State Assign the first vehicle $v_1$ in $\mathcal{V}_{un}$ to the RSU which is closer to it. Update $\xi_{[i,1],n}$, and update vehicle $v_1$ to $\mathcal{V}$, then remove $v_1$ from $\mathcal{V}_{un}$.
\Repeat
\State Select the first vehicle $v_k$ in $\mathcal{V}_{un}$. 
\State Initialize $\mathbf{F}_n$ such that $\mathbf{f}_{[i,k],n}=\mathbf{a}(\varphi_{[i,k],n})$.
\Repeat
\State Update $\mathbf{\gamma}_n$ by (\ref{gamma}).
\State Update $\mathbf{Y}_n$ by (\ref{x}).
\State Find the optimal $\mathbf{F}_n$ using PGD algorithm.
\Until{the value of $f_q(\mathbf{\xi}_n,\mathbf{F}_n,\mathbf{\gamma}_n,\mathbf{Y}_n)$ in  (\ref{fq}) converges.}
\State Compare $\zeta_{[1,k],n}$ and $\zeta_{[2,k],n}$ with respect to the vehicles set $\mathcal{V}$.
\If{$\zeta_{[1,k],n} > \zeta_{[2,k],n} $}
\State Set $\xi_{[1,k],n}=1$ and $\xi_{[2,k],n}=0$.
\Else 
\State Set $\xi_{[1,k],n}=0$ and $\xi_{[2,k],n}=1$.
\EndIf
\State Update $v_k$ to $\mathcal{V}$, then remove it from $\mathcal{V}_{un}$.
\Until{$\mathcal{V}_{un}$ is empty.}
\end{algorithmic}
\label{Greedy}
\end{algorithm}

\subsubsection{Greedy Algorithm}

In this algorithm, vehicles are allocated to RSUs one by one. The algorithm starts with the vehicle that exhibits the largest difference in distances between itself and both RSUs, considering this vehicle suffers from the greatest impact of path loss. Let $\mathcal{V}$ represent the set of vehicles in the system. The greedy algorithm is shown in Algorithm \ref{Greedy}. Through this approach, a reasonable assignment strategy can be established by assigning the vehicles one by one.

\subsubsection{Heuristic Algorithm}

We formulate the following optimization problem to optimize $\mathbf{\xi}_n$: 

\begin{equation}
\begin{aligned}
    \max_{\mathbf{\xi}_n} & \sum_{i=1}^I\sum_{k=1}^K \xi_{[i,k],n}\zeta_{[i,k],n}\\
    \text{subject to}\:\: &  \xi_{[i,k],n}\in \{0,1\},\forall k,\forall i\\
    &\sum_{i=1}^{I}\xi_{[i,k],n}=1,\forall k.\\
\end{aligned}
\label{Hg}
\end{equation}

\begin{algorithm}[t!]
\caption{Heuristic Optimization Algorithm}
\begin{algorithmic}[1]
\State Initialize $\mathbf{\xi}_n$ based on the distances between each vehicle and both RSUs and initialize $\mathbf{F}_n$ such that $\mathbf{f}_{[i,k],n}=\mathbf{a}(\varphi_{[i,k],n})$.
\State Compute $\zeta_{[i,k],n}, \forall i \in \{1,2\}, \forall k \in [1,K]$.
\State Compute $e_{k,n} = \sum_{i=1}^{2}\zeta_{[i,k],n}(1-2\xi_{[i,k],n}),\forall k \in [1,K]$.
\While {$\exists k, e_{k,n} > 0$}
\State Select $k^*$ that gives the maximum $e_{k,n}$.
\State Update $\xi_{[i,k],n} \leftarrow 1-\xi_{[i,k],n}$
\Repeat
\State Update $\mathbf{\gamma}_n$ by (\ref{gamma}).
\State Update $\mathbf{Y}_n$ by (\ref{x}).
\State Find the optimal $\mathbf{F}_n$ using PGD algorithm.
\Until{The value of $f_q(\mathbf{\xi}_n,\mathbf{F}_n,\mathbf{\gamma}_n,\mathbf{Y}_n)$ converges.}
\State Update $\zeta_{[i,k],n}, \forall i \in \{1,2\}, \forall k \in [1,K]$.
\State Update $e_{k,n}, \forall k \in [1,K]$.
\EndWhile
\end{algorithmic}
\label{heuristic}
\end{algorithm}

This optimization problem is challenging to solve as every single change in $\mathbf{\xi}_{n}$ results in a series of changes in $\zeta_{[i,k],n}$. Here, we employ a heuristic algorithm to find the optimal matching between the vehicles and the RSUs, as shown by Algorithm \ref{heuristic}. In each iteration, we select the vehicle that is most suitable for the other RSU and assign it to that RSU. The suitability of a vehicle for a particular RSU can be reflected by $\zeta_{[i,k],n}$. This algorithm terminates when all vehicles are assigned to the most suitable RSU.

\begin{lemma}
$f_q(\mathbf{\xi}_n,\mathbf{F}_n,\mathbf{\gamma}_n,\mathbf{Y}_n)$ converges to a stationary point.
\end{lemma}
\begin{proof}
The aim of (\ref{gamma}) and (\ref{x}) is to find the parameters that maximizes $f_q(\mathbf{\xi}_n,\mathbf{F}_n,\mathbf{\gamma}_n,\mathbf{Y}_n)$, hence $f_q(\mathbf{\xi}_n,\mathbf{F}_n,\mathbf{\gamma}_n,\mathbf{Y}_n)$ is non-decreasing after each iteration.
\end{proof}

\section{Learning-based Beamforming Design and Vehicle Assignment}

Solving the optimization problem proves to be time-consuming, because of the heuristic algorithm employed and the high computational complexity associated with each iteration. This work also introduces an RNN-based technique for beamforming design and vehicle assignment.
RNNs are a class of neural networks that are well-suited for modeling sequential data. As vehicle allocations and beamforming matrices are significantly influenced by the adjacent vehicles, they have to be optimized jointly. The impact of neighboring vehicles is managed by the internal states iterated within the RNN layer. Additionally, the number of vehicles to be allocated within the system varies over time, leading to time-varying input dimensions to the learning algorithm. RNNs can leverage their internal states to handle variable-length sequences of inputs, making them ideal for our requirements.

\subsection{LSTM Network}

LSTM networks, introduced by Hochreiter and Schmidhuber \cite{LSTM}, extend RNNs by providing short-term memory mechanisms to preserve internal states. A standard LSTM cell includes one input gate, one output gate, and one forget gate. The output of the forget gate can be calculated by
\begin{equation}
    f_t = \sigma (W_f\cdot [h_{t-1},x_t]+b_f),
\end{equation}
where $f_t$ lies in the range $(0,1)$, $W_f$ is the weight of the forget gate, $b_f$ is the bias of the forget gate, $x_t$ is the current input, and $h_{t-1}$ is the previous output value.

Similarly, the output of the input gate and the output gate can be respectively written as
\begin{equation}
    i_t = \sigma (W_i \cdot [h_{t-1},x_t]+b_i),
\end{equation}
\begin{equation}
    o_t = \sigma (W_o  \cdot [h_{t-1},x_t]+b_o),
\end{equation}
where $W_i$ and $W_o$ are the weights of the input gate and output gate respectively, $b_i$ and $b_o$ are the bias of the input gate and output gate respectively. While LSTM networks are commonly employed in time series contexts, we have found them effective in our model due to the robust correlation among adjacent users after sorting the vehicles based on their locations. The key advantage of LSTMs over traditional RNNs lies in their ability to learn and remember long-term dependencies in sequential data, due to their gated architecture and memory cells. This makes them well-suited for our work where the beamforming matrices of multiple vehicles need to be jointly considered.

\subsection{Bi-LSTM-based Vehicle Assignment and Beamforming Design}

Since a vehicle is affected by the vehicles on both sides, we employ a bi-directional LSTM framework, as depicted in Fig.~\ref{RNN2}, which encompasses $K$ sets of input features. Organizing information about the vehicles based on their positions before inputting into the learning algorithm is crucial. Considering the real and imaginary parts individually, the input feature of the neural network includes $\mathrm{Re}\{\mathbf{r}_{[i,k],n}\}$, $\mathrm{Im}\{\mathbf{r}_{[i,k],n}\}$, $\mathrm{Re}\{\mathbf{f}_{[i,k],n-1}\}$, $\mathrm{Im}\{\mathbf{f}_{[i,k],n-1}\}$, $\tilde{\nu}_{[i,k],n}$, and $\tilde{\mu}_{[i,k],n}$. Given the multi-dimensional nature of the input, a flatten layer follows each input layer to generate a one-dimensional tensor. 

\begin{figure}[!t]
    \centering
    \includegraphics[width=0.95\columnwidth]{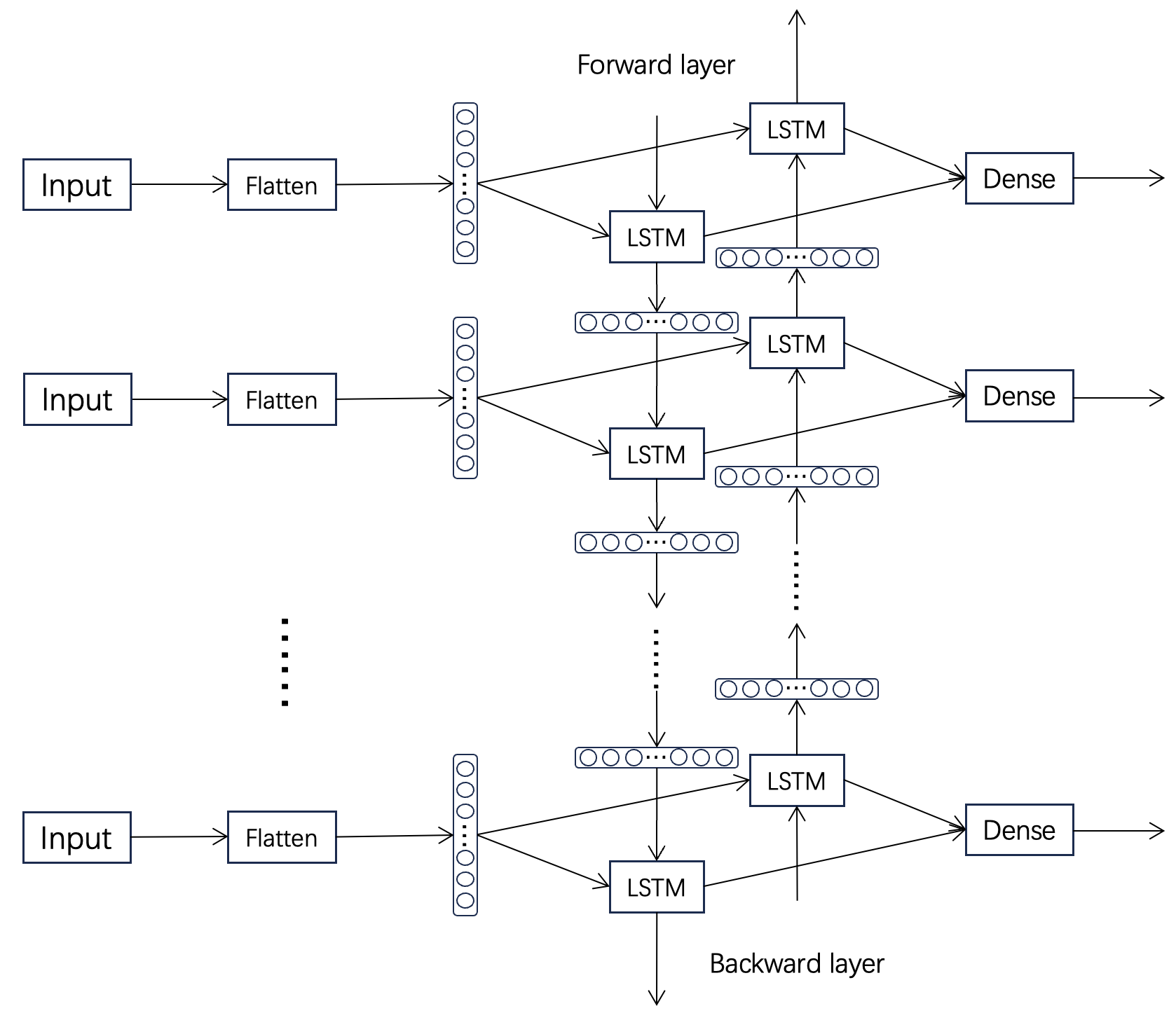}
    \caption{The bi-directional convolutional LSTM architecture for beamforming optimization.}
    \label{RNN2}
\end{figure}

The LSTM output is then forwarded to the adjacent LSTM to leverage the correlation between the beamforming of adjacent vehicles. Subsequently, fully connected dense layers, with Rectified Linear Unit (ReLU) serving as the activation function, follow the LSTM layer to generate the output. The output includes both the real and imaginary parts of the beamforming, i.e., $[\mathrm{Re}\{\mathbf{f}_{[i,k],n}\}, \mathrm{Im}\{\mathbf{f}_{[i,k],n}\}] \in \mathbb{C}^{2N_t\times 1}$, used to formulate the predictive beamforming. MSE is employed as the loss function for evaluating the predicted beamforming. Notably, the overfitting of this network is negligible, hence drop-out layers are not needed.

\subsection{Complexity Analysis}

The overall complexity of the proposed algorithm is the summation of the complexity of two separate learning algorithms, both are bi-directional LSTM networks. LSTM is local in both space and time, which means that the storage requirements of the network are irrelevant to the actual input size. For each time slot, the time complexity of each weight of the LSTM network equals $\mathcal{O}(1)$.

The weight of the LSTM network can be computed by $W=4[h(h+e)+h]$ with $h$ and $e$ being the number of hidden units and the embedding dimension of input. Since bi-directional LSTM is applied in both algorithms, the total time complexity of the proposed LSTM framework can be written as
\begin{equation}
    C_{1} = \mathcal{O}\left(8[h_1(h_1+e_1)+h_1]\right)= \mathcal{O}\left(8[h_1(h_1+e_1)]\right),
\end{equation}
where $h_l$ and $e_l$ are the number of hidden units and the embedding dimension of input. 

The complexity of each iteration in the greedy algorithm $C_{g}$ and the heuristic algorithm $C_h$ can be expressed by
\begin{equation}
    C_{g} = C_{h} = \mathcal{O}\left(K^2 N_t + \frac{IK^2N_t^2}{\epsilon}\right)=\mathcal{O}\left( \frac{IK^2N_t^2}{\epsilon}\right),
\end{equation}
where $\epsilon$ denotes the accuracy. Normally, the heuristic algorithm is faster when $K$ is low and the greedy algorithm is faster when $K$ is high. The complexity of the LSTM network has lower degrees than the heuristic algorithm, hence the LSTM is generally less complex especially when the number of vehicles is high.

\section{Numerical Results}

\begin{table*}[!t]
\centering
\caption{The settings in the simulations}
\begin{tabularx}{0.7\textwidth}{ 
   >{\centering\arraybackslash}X 
  >{\centering\arraybackslash}X  }
  \hline
Parameter                                               & Value                           \\ \hline
Complex fading coefficient                              & $\varrho = 10+10j$              \\
Channel power gain at $d_0 = 1$m & $\tilde{\alpha}=-70$dB          \\
Sensing channel noise variance& $\sigma_e^2=-70$dB          \\
Communication channel noise variance& $\sigma_c^2=-70$dB          \\
Index corresponds to $\sigma_{\mathbf{r},[i,k],n}^2$    & $\rho_{\mathbf{r}}=1$           \\
Index corresponds to $\sigma_{\nu,[i,k],n}^2$           & $\rho_\nu = 6.7 \times 10^{-7}$ \\
Index corresponds to $\sigma_{\mu,[i,k],n}^2$           & $\rho_\mu = 2 \times 10^4$      \\
Carrier frequency                                       & $f_c=30$ GHz                    \\
Slot duration                                           & $T=0.02$ s                      \\
Matched-filter gain                                     & $G = 10$                        \\
Location of RSU1                                        & $[0 \text{m},30 \text{m}]$                    \\
Location of RSU2                                        & $[0\text{m},-30\text{m}]$                     \\
Center of the highway                          & $[-30\text{m},0\text{m}] \to [30\text{m},0\text{m}]$                          \\
Width of the highway & 10m\\

\hline
\end{tabularx}
\label{Setting}
\end{table*}

We conducted various simulations in this section to evaluate the performance of our proposed methods. The default settings of the simulation parameters are shown in Table \ref{Setting}. We first set $N_r=N_t=32$ and train the LSTM networks with a training set of 10,000 samples and a validation set of 3,000 samples. The training results in terms of the accuracy rate for the vehicle assignment problem and the MSE for the beamforming optimization problem are shown in Fig. \ref{Accuracy} and Fig. \ref{MSE}, respectively. We can achieve an accuracy rate of above 96\% after tens of training epochs for vehicle assignment. The LSTM network for beamforming optimization requires more epochs for training. The optimal performance can be achieved after roughly 200 training epochs unless the number of vehicles in the system is too large. Note that as the number of vehicles increases, the accuracy rate decreases, while the MSE of the predicted beamforming also decreases. This reduction in MSE occurs because, with more vehicles in the system, distribution tends to be more even, reducing randomness. Conversely, when fewer vehicles are present, the randomness of the system increases significantly, which complicates beamforming prediction. Furthermore, with only a few vehicles, there typically is one optimal assignment, with other assignments being significantly inferior. However, as the number of vehicles increases, the system accommodates more sub-optimal solutions, which helps maintain relatively good performance despite the decrease in accuracy rate.

\begin{figure}[!t]
    \centering   \includegraphics[width=0.8\columnwidth]{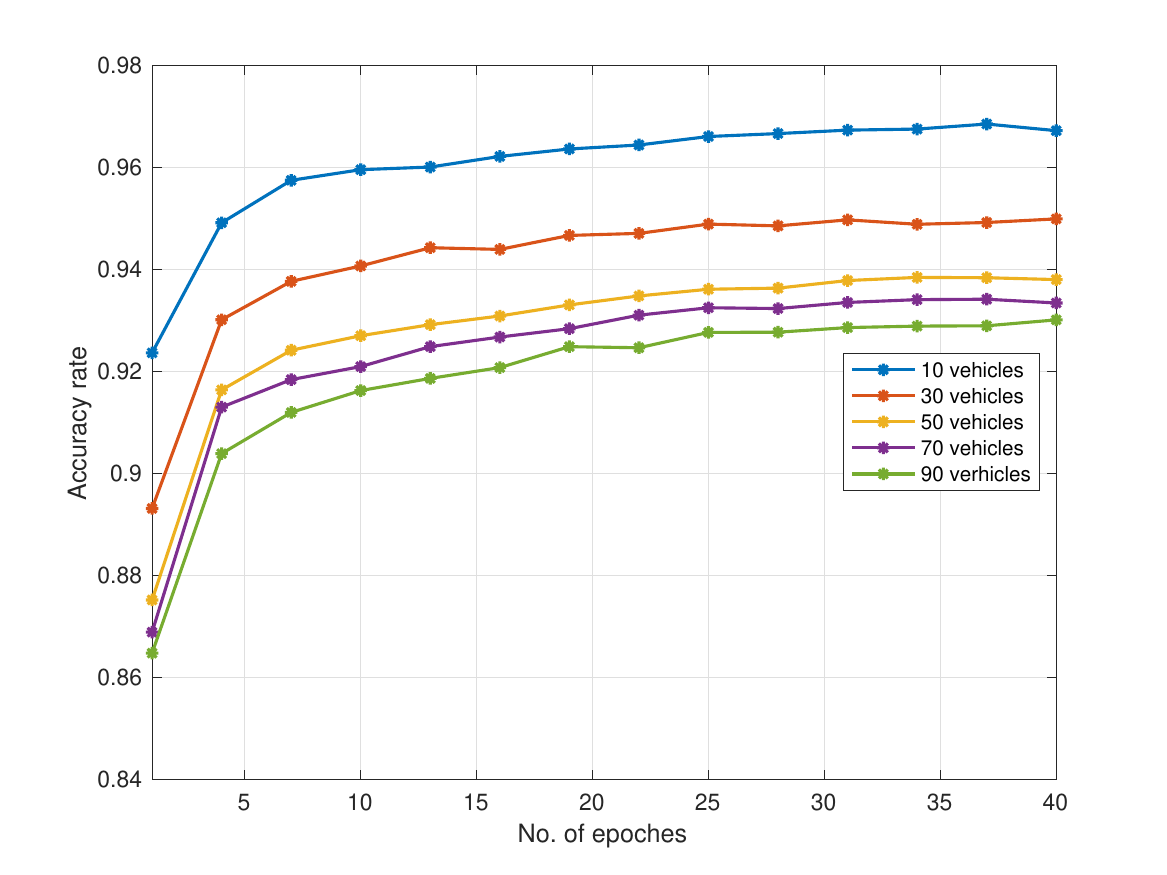}
    \caption{The evolution of the accuracy rate after different numbers of training epochs.}
    \label{Accuracy}
\end{figure}

\begin{figure}[!t]
    \centering   \includegraphics[width=0.8\columnwidth]{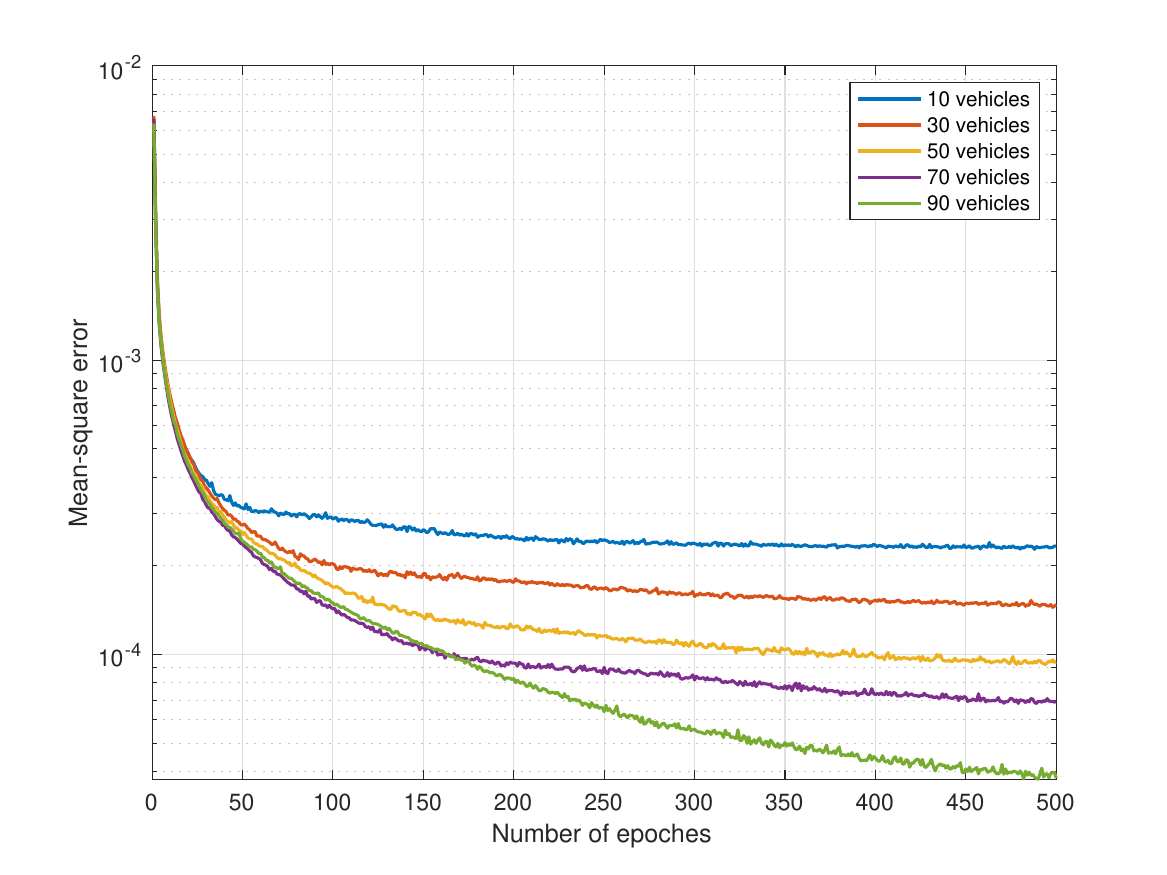}
    \caption{The evolution of the MSE of the output result after different numbers of training epochs.}
    \label{MSE}
\end{figure}

\begin{figure}[!t]
    \centering   \includegraphics[width=0.8\columnwidth]{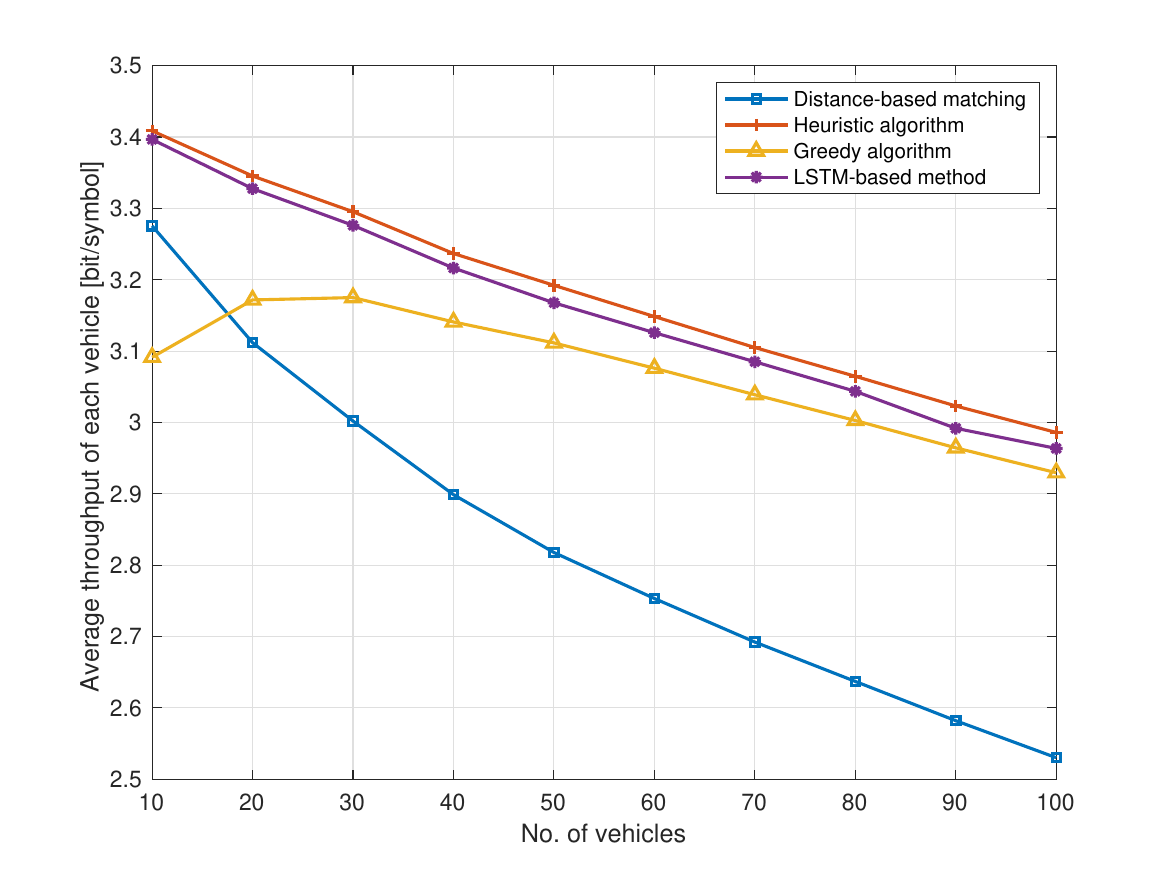}
    \caption{The average throughput of each vehicle with different numbers of vehicles in the system when assuming $N_t=N_r=32$.}
    \label{Th vs. ve}
\end{figure}

With the number of antennas being constant, we change the number of vehicles in the system. We compare the performance of the optimal matching and beamforming obtained from the heuristic algorithm, the LSTM-based method, the Greedy method, and the conventional distance-based vehicle-RSU matching. Fig. \ref{Th vs. ve} presents the relationship between the number of vehicles in the system and the average throughput of each vehicle. Generally, the average throughput of all four schemes decreases sharply as the number of vehicles increases. This decrease occurs because more vehicles introduce more interference and complicate the assignment. Our proposed DT-based vehicle assignment and beamforming design significantly outperforms the traditional distance-based method, especially when more vehicles are served simultaneously. However, the greedy algorithm performs even worse than the distance-based method when the number of vehicles is low but approaches the performance of Algorithm \ref{heuristic} when the number of vehicles is high. This behavior is due to the tendency of the greedy algorithm to produce sub-optimal matchings. When there are only a few vehicles in the system, the optimal matching between vehicles and RSUs is unique, and sub-optimal matching is far worse than the optimal one. However, as the number of vehicles increases, this performance gap becomes smaller, and there may even be multiple optimal matching schemes. Furthermore, we can see that with the LSTM network, vehicles can achieve throughput only slightly lower than it obtained from Algorithm \ref{heuristic}, while significantly reducing processing time.

\begin{figure}[!t]
    \centering   \includegraphics[width=0.8\columnwidth]{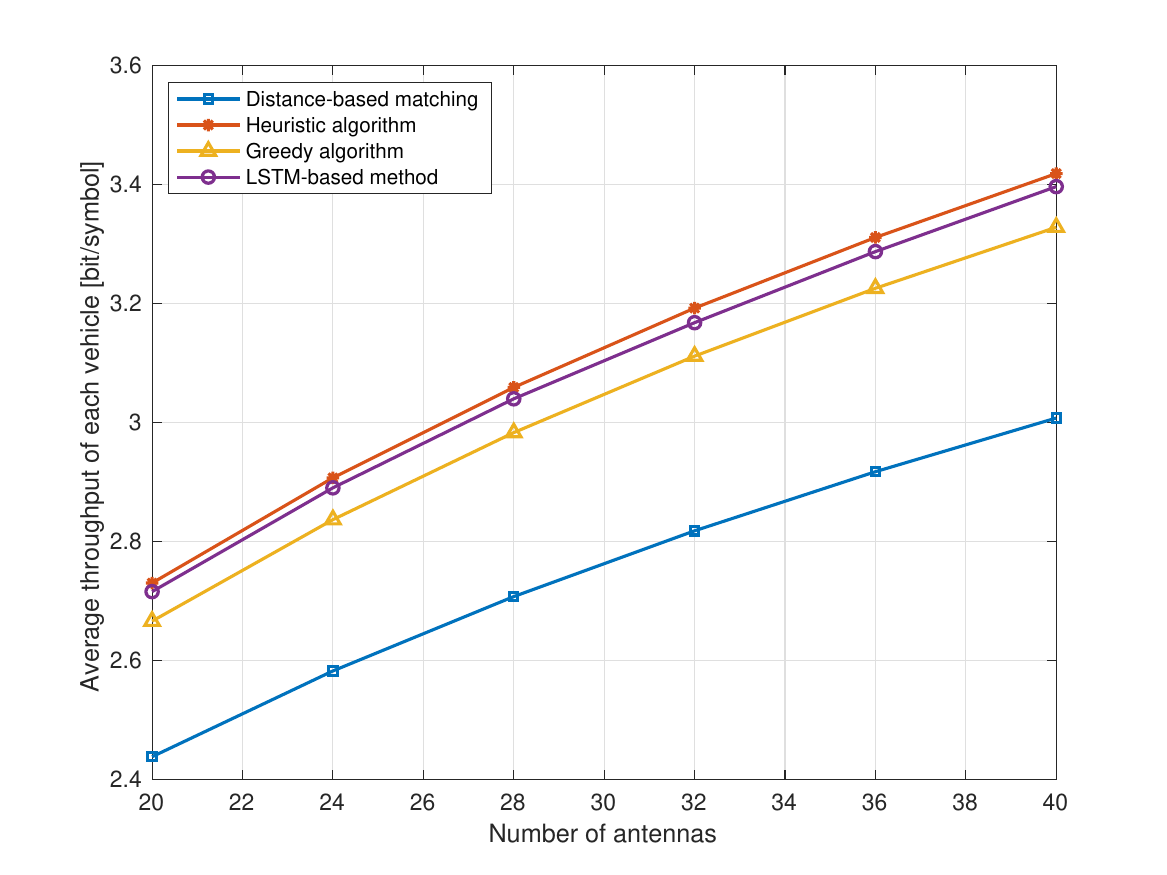}
    \caption{The average throughput of each vehicle with different numbers of transmit/receive antennas at the RSUs, when the number of vehicles is fixed and set as $K=50$.}
    \label{antenna}
\end{figure}

Apart from the number of vehicles, the performance of the system is also influenced by the number of antennas at the RSUs. We keep the number of vehicles in the system constant and vary the number of antennas, as shown in Fig. \ref{antenna}. It is assumed that the number of transmit antennas and the number of receive antennas are equal at the RSUs and there are $K=50$ vehicles. Clearly, the average throughput increases when there are more antennas. The results in the figure demonstrate that Algorithm \ref{heuristic} achieves the best performance, with the LSTM slightly behind, followed by the greedy algorithm and the distance-based method.

In all the previous simulations, we assume a constant number of vehicles. However, in reality, vehicles continuously enter and exit the system, leading to fluctuations in the number of vehicles. To see how the communication performance and sensing performance varies over time in a dynamic environment for a continuous period, we compare these schemes while assuming that vehicle arrivals follow a Poisson process. Given that the vehicles are on a highway, we assume that their instantaneous speeds follow a Gaussian distribution with a mean of 30 meters per second, and we ignore any congestion on the highway. To estimate the average number of vehicles when the system is in a stationary state, we apply Little's law. Since the arrival of vehicles is random, the results of each simulation vary. Therefore, we simulate the system for lots of times and used one of the most common cases as an example.

\begin{figure}[!t]
    \centering   \includegraphics[width=0.8\columnwidth]{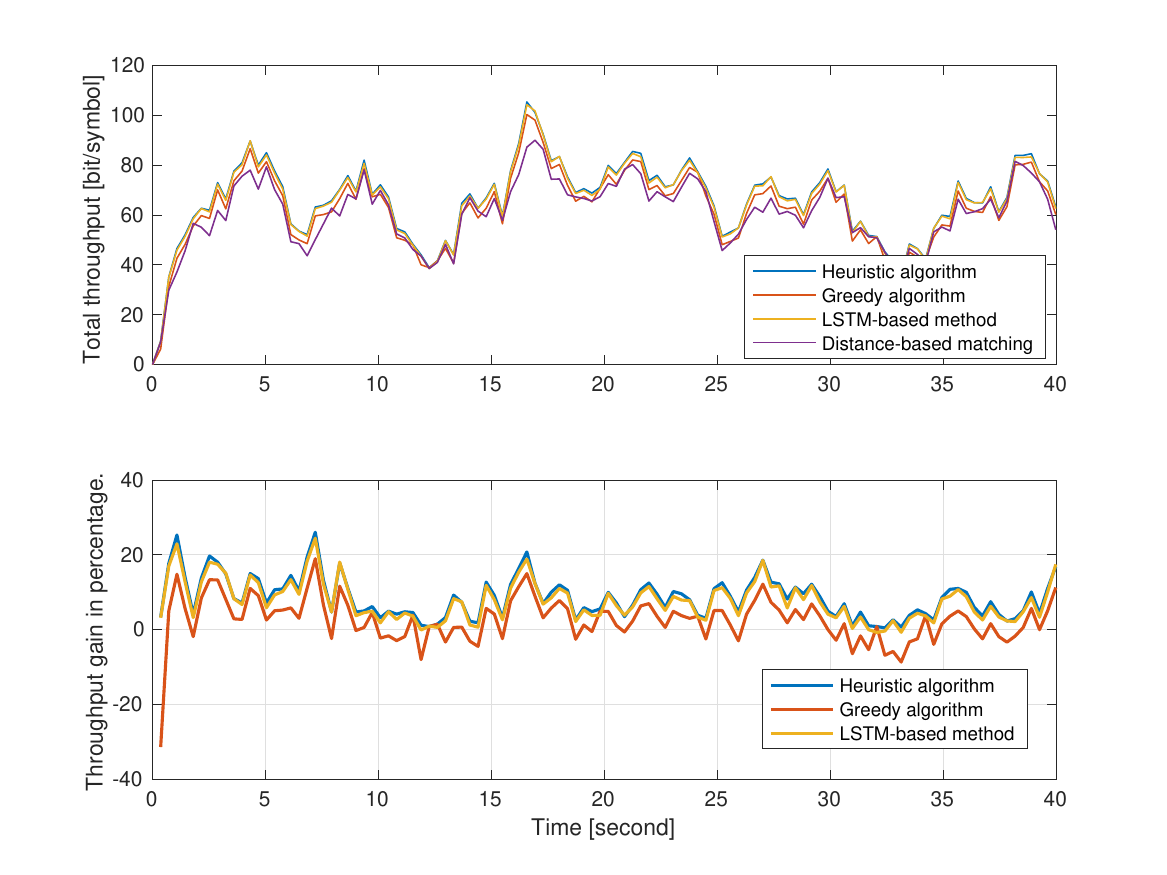}
    \caption{An example of the total instantaneous throughout with an average arrival rate of 5 vehicles per second in each direction.}
    \label{Poisson 20}
\end{figure}

Initially, we set the average arrival rates in both directions to 5 vehicles per second, resulting in an average of 20 vehicles in the system. Starting from an initial state where the system is vacant, we continuously monitor performance for 40 seconds and present the simulation results in Fig. \ref{Poisson 20}. The upper half of Fig. \ref{Poisson 20} displays the total throughput comparison of all four schemes. It shows that the total throughput increases initially and then fluctuates around a stationary level. It is challenging to point out the superior scheme on the upper half of Fig. \ref{Poisson 20}. To provide more clarity, we present the throughput gain of the other three schemes relative to the distance-based matching on the lower half of Fig. \ref{Poisson 20}, which corresponds to the upper half exactly. In the first three seconds, the throughput gain of the greedy algorithm remains negative, which also demonstrates the results in Fig. \ref{Th vs. ve}. As the number of vehicles increases, the greedy algorithm generally outperforms the distance-based matching, but it exhibits occasional negative throughput gains, indicating its high instability. The matching based on Algorithm \ref{heuristic} consistently demonstrates superiority, except for scenarios with very few vehicles, where it achieves zero gain. The LSTM-based method, while slightly inferior to Algorithm \ref{heuristic}, still outperforms the other schemes in most cases.

\begin{figure}[!t]
    \centering   \includegraphics[width=0.8\columnwidth]{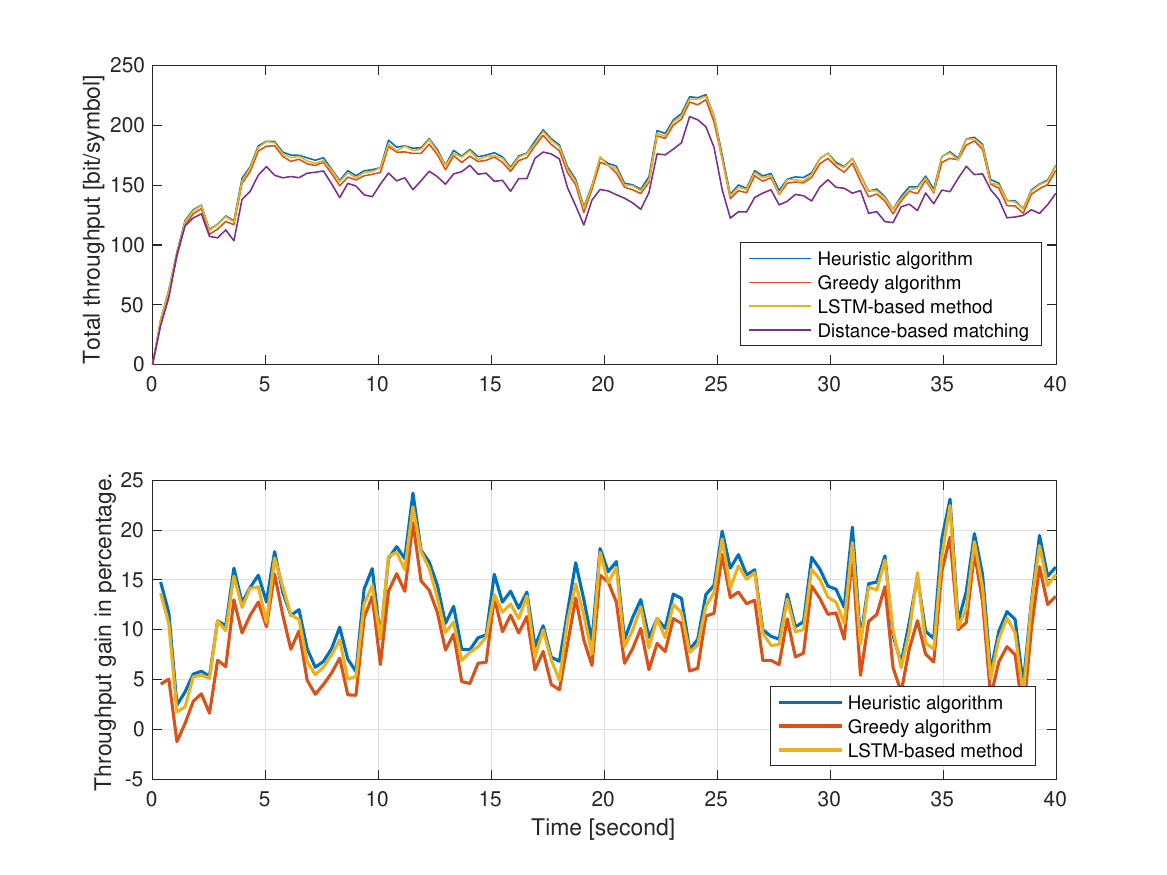}
    \caption{An example of the total instantaneous throughout with an average arrival rate of 12.5 vehicles per second in each direction.}
    \label{Poisson 50}
\end{figure}

Fig. \ref{Poisson 50} displays the results when we increase the arrival rate to 12.5 vehicles per second, resulting in an average of 50 vehicles when the system is stationary. Similar to the previous scenario, the total throughput initially increases and then fluctuates around a stationary value, but the initial increase is sharper. As the number of vehicles rapidly rises to a relatively high level, the throughput gain of the greedy algorithm is negative for only one to two seconds and remains positive for the remainder of the time. The overall performance ranking remains consistent: Algorithm \ref{heuristic} is the best, followed by the LSTM-based method and the greedy algorithm, all of which outperform the distance-based matching. Additionally, we observe that the stability of the greedy algorithm significantly improves as the average number of vehicles in the system increases. Furthermore, in both Fig. \ref{Poisson 20} and Fig. \ref{Poisson 50}, the total throughput increases initially and remains high as the system stabilizes, indicating that the sensing performance also remains high. This consistent level of throughput suggests that the sensing accuracy of the system is effectively supporting the communication processes, corroborating our approach that links high communication rates with robust sensing performance.

\begin{figure}[!t]
    \centering   \includegraphics[width=0.8\columnwidth]{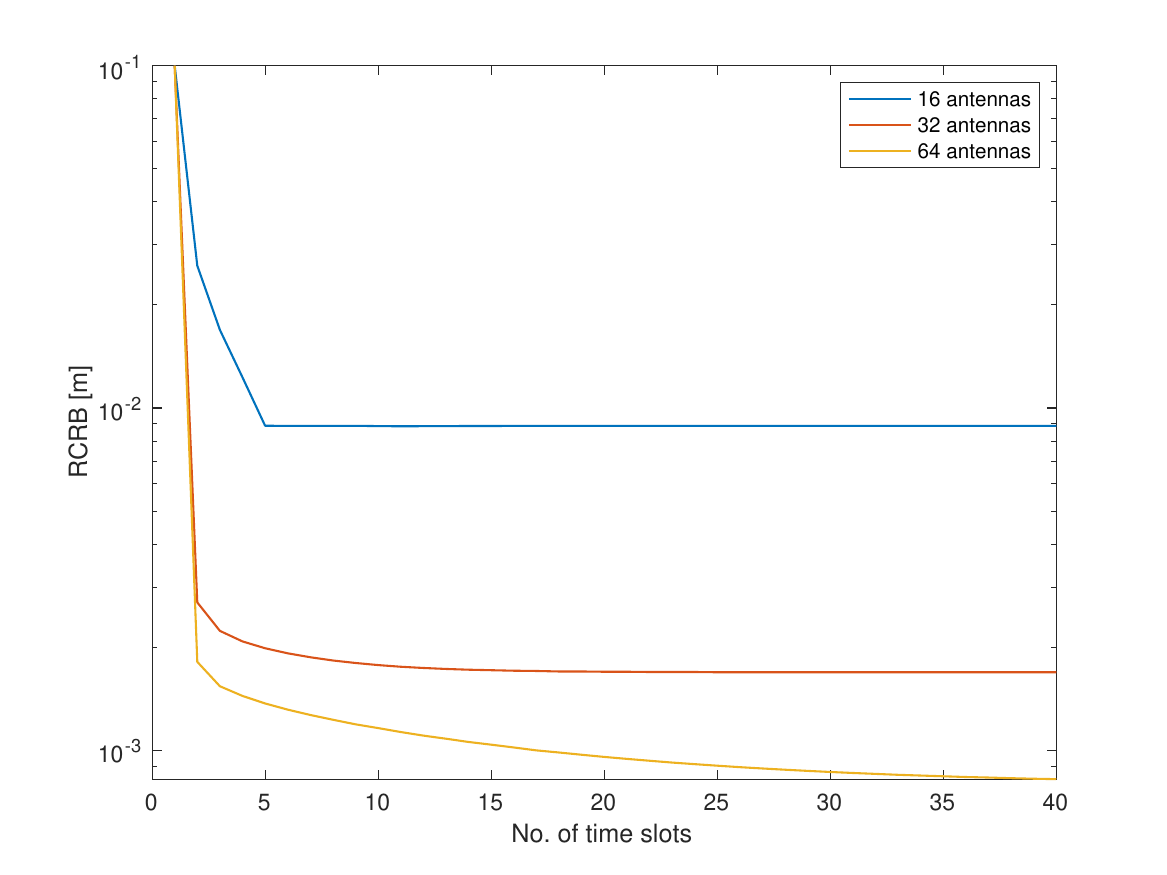}
    \caption{The evolution of the RCRB with different numbers of antennas.}
    \label{RCRB}
\end{figure}

In Fig. \ref{RCRB}, we compare the root of CRB (RCRB) for scenarios with 16, 32, and 64 transmit antennas. Initially setting RCRB to 0.1 meters, our results demonstrate a sharp decrease in RCRB at the beginning due to the design of beamforming techniques that enhance both communication and sensing performances. Subsequently, RCRB reaches a steady state where it converges to a constant value suitable for high-quality communication purposes. Notably, as the number of antennas increases, the RCRB bound decreases.

According to the simulation results, our proposed vehicle assignment and beamforming optimization algorithm outperforms all the other schemes no matter whether there are more or fewer vehicles in the system. The LSTM-based algorithm brings us slightly lower throughput but significantly reduces the processing time. The greedy algorithm performs the worst when the number of vehicles in the system is low, but approaches the optimal method when the number of vehicles increases.   

\section{Conclusion}

In this work, we introduce a design for DT-based ISAC in vehicular networks. The physical system in the real world comprises two RSUs and multiple vehicles. The DT is used to create the motion model of the vehicles to track them and perform intricate computations aimed at optimizing the overall communication rate. The received sensing information from the vehicles is then forwarded to the DT for beamforming design, vehicle assignments, and potentially other applications. While the optimization problem can be solved iteratively, we also leverage bi-directional LSTM networks to expedite decision-making. Our results demonstrate a substantial improvement in overall throughput. Furthermore, this work has the potential for extension to scenarios involving multiple RSUs and applications in facilitating handovers.


\appendix
\section{Find R}\label{appendix:R}

First, we define $\mathbf{\Omega}_{[i,k],n}$ as
\begin{equation}
\begin{aligned}
    \mathbf{\Omega}_{[i,k],n} = &\mathbf{H}_{[i,k],n}\mathbf{Q}_{[i,k],n}^{-1}\mathbf{H}_{[i,k],n}^H\\
    &+ (\mathbf{G}_{[i,k],n} \tilde{\mathbf{M}}_{[i,k],n-1}\mathbf{G}_{[i,k],n}^H)^{-1}.
\end{aligned}
\end{equation}
 
Then, $\tilde{m}_{[i,k],n}^{(11)} \leq \tilde{m}_{[i,k],n-1}^{(11)}$ is equivalent to
\begin{equation}
\begin{aligned}
    &(\omega_{(2,2)}\omega_{(3,3)}-\omega_{(2,3)}\omega_{(3,2)})/\Big(\omega_{(1,1)}(\omega_{(2,2)}\omega_{(3,3)}\\
    &-\omega_{(2,3)}\omega_{(3,2)})-\omega_{(1,2)}(\omega_{(2,1)}\omega_{(3,3)}-\omega_{(2,3)}\omega_{(3,1)})\\
    &+\omega_{(1,3)}(\omega_{(2,1)}\omega_{(3,2)}-\omega_{(2,2)}\omega_{(3,1)})\Big)\leq \tilde{m}_{[i,k],n-1}^{(11)},
\end{aligned}
\end{equation}
where $\omega^{(i,j)}$ is the $(i,j)$-th entry of $\mathbf{\Omega}_{[i,k],n}$. Because $\mathbf{\Omega}_{[i,k],n}$ is positive-definite, we have
\begin{equation}
\begin{aligned}
    &\Big(\omega_{(1,2)}(\omega_{(2,1)}\omega_{(3,3)}-\omega_{(2,3)}\omega_{(3,1)}) + \omega_{(1,3)}(\omega_{(2,1)}\omega_{(3,2)}\\
    &-\omega_{(2,2)}\omega_{(3,1)}) + (\omega_{(2,2)}\omega_{(3,3)}-\omega_{(2,3)}\omega_{(3,2)})/\tilde{m}_{[i,k],n-1}^{(11)}\Big)\\
    &/(\omega_{(2,2)}\omega_{(3,3)}-\omega_{(2,3)}\omega_{(3,2)})\leq \omega_{(1,1)}.
\end{aligned}
\label{56}
\end{equation}

We define $\hat{m}_{[i,k],n}^{(1,1)}$ as the $(1,1)$-th entry of $\hat{\mathbf{M}}_{[i,k],n}$, then $\omega_{(1,1)}$ can be written as
\begin{equation}
\begin{aligned}
    &\omega_{(1,1)} = \\
    &\frac{\kappa^2 \beta_{[i,k],n}^2 G^2}{\sigma_{\mathbf{r},[i,k],n}^2}\left(\mathbf{\Pi}_{[i,k],n}\mathbf{f}_{[i,k],n}\right)^H \left(\mathbf{\Pi}_{[i,k],n}\mathbf{f}_{[i,k],n}\right) + \hat{m}_{[i,k],n}^{(1,1)}.
\end{aligned}
\label{57}
\end{equation}

Based on (\ref{56}) and (\ref{57}), $\Lambda_{[i,k],n}$ in (\ref{Constraint}) can be calculated by
\begin{equation}
\begin{aligned}
    \Lambda_{[i,k],n} &= \Bigg(\Big(\omega_{(1,2)}(\omega_{(2,1)}\omega_{(3,3)}-\omega_{(2,3)}\omega_{(3,1)}) \\
    &+ \omega_{(1,3)}(\omega_{(2,1)}\omega_{(3,2)}-\omega_{(2,2)}\omega_{(3,1)}) + (\omega_{(2,2)}\omega_{(3,3)}\\
    &-\omega_{(2,3)}\omega_{(3,2)})/\tilde{m}_{[i,k],n-1}^{(11)}\Big)/(\omega_{(2,2)}\omega_{(3,3)}\\
    &-\omega_{(2,3)}\omega_{(3,2)})-\hat{m}_{[i,k],n}^{(1,1)} \Bigg)\frac{\sigma_{\mathbf{r},[i,k],n}^2}{\kappa^2 \beta_{[i,k],n}^2 G^2}.
\end{aligned}
\end{equation}

\bibliographystyle{IEEEtran}
\bibliography{MMM}

\begin{IEEEbiography}[{\includegraphics[width=1in,height=1.25in,clip,keepaspectratio]{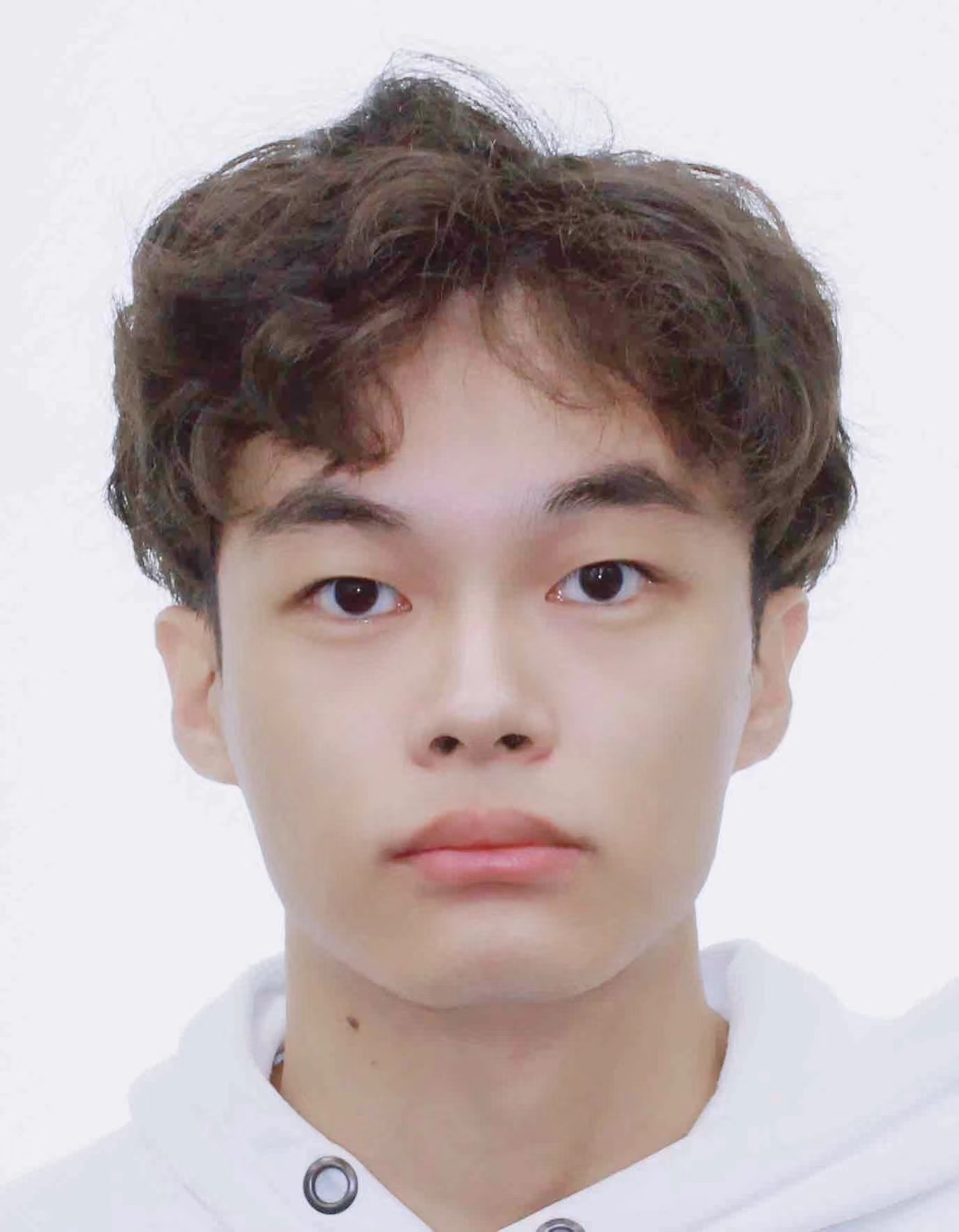}}]{Weihang Ding}
received the B.Sc. degree in telecommunication engineering from Beijing Jiaotong University, Beijing, China, in 2018, the M.Sc. degree in telecommunications from King's College London, London, United Kingdom, in 2019. He is currently pursuing the Ph.D. degree in telecommunications at King's College London too. His main research interests include UAV communication, physical layer resource allocation, hybrid automatic repeat requests, channel coding, and integrated sensing and communication.
\end{IEEEbiography}

\begin{IEEEbiography}[{\includegraphics[width=1in,height=1.25in,clip,keepaspectratio]{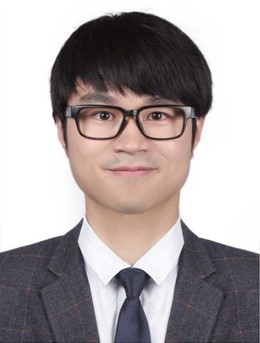}}]{Zhaohui Yang}
(Member, IEEE) is currently a ZJU Young Professor with the Zhejiang Key Laboratory of Information Processing Communication and Networking, College of Information Science and Electronic Engineering, Zhejiang University. He received the Ph.D. degree from Southeast University, Nanjing, China, in 2018. From 2018 to 2020, he was a Post-Doctoral Research Associate at the Center for Telecommunications Research, Department of Informatics, King’s College London, U.K. From 2020 to 2022, he was a Research Fellow at the Department of Electronic and Electrical Engineering, University College London, U.K. His research interests include joint communication, sensing, and computation, federated learning, and semantic communication. He received IEEE Communications Society Leonard G. Abraham Prize Award in 2024, IEEE Marconi Prize Paper Award in 2023, IEEE Katherine Johnson Young Author Paper Award in 2023. He currently serves as an Associate Editor of IEEE TGCN, IEEE CL, IEEE TMLCN. He has served as a Guest Editor for several journals including IEEE \textsc{Journal on Selected Areas in Communications}.
\end{IEEEbiography}

\begin{IEEEbiography}[{\includegraphics[width=1in,height=1.25in,clip,keepaspectratio]{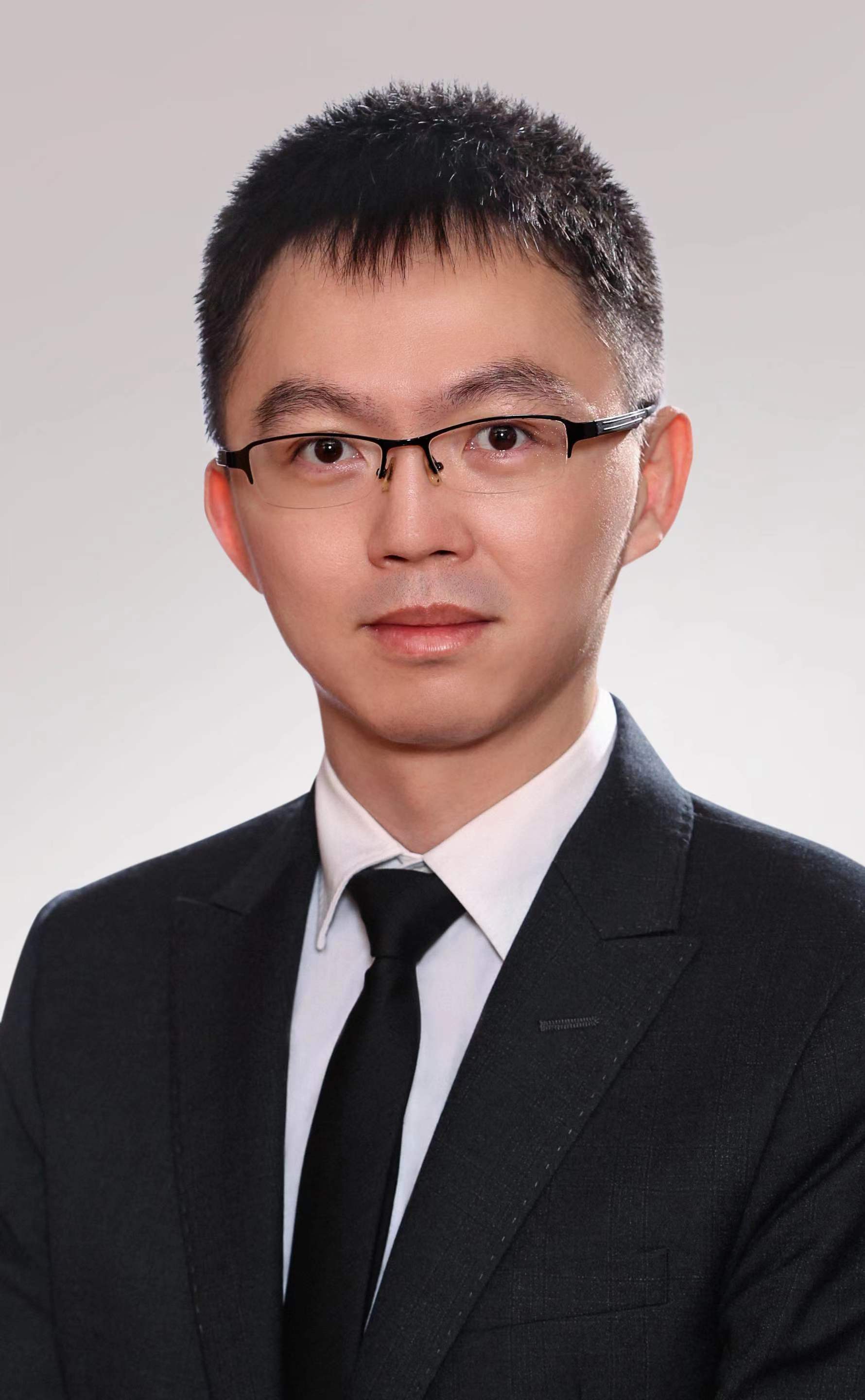}}]{Mingzhe Chen}
is currently an Assistant Professor with the Department of Electrical and Computer Engineering and Frost Institute of Data Science and Computing at University of Miami. His research interests include federated learning, reinforcement learning, virtual reality, unmanned aerial vehicles, and Internet of Things. He has received four IEEE Communication Society journal paper awards including the IEEE Marconi Prize Paper Award in Wireless Communications in 2023, the Young Author Best Paper Award in 2021 and 2023, and the Fred W. Ellersick Prize Award in 2022, and four conference best paper awards at ICCCN in 2023, IEEE WCNC in 2021, IEEE ICC in 2020, and IEEE GLOBECOM in 2020. He currently serves as an Associate Editor of IEEE Transactions on Mobile Computing, IEEE Transactions on Communications, IEEE Wireless Communications Letters, IEEE Transactions on Green Communications and Networking, and IEEE Transactions on Machine Learning in Communications and Networking.
\end{IEEEbiography}

\begin{IEEEbiography}[{\includegraphics[width=1in,height=1.25in,clip,keepaspectratio]{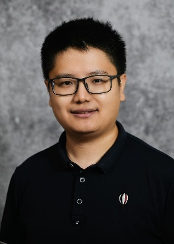}}]{Yuchen Liu}Yuchen Liu is currently an Assistant Professor with the Department of Computer Science at North Carolina State University, USA. He got his Ph.D. degree at the Georgia Institute of Technology, USA. His research interests include wireless networking, digital twins, generative AI, distributed learning, mobile computing, and software simulation. He has received several best paper awards at IEEE and ACM conferences. He currently serves as Associate Editors of IEEE Transactions on Green Communications and Networking, IEEE Transactions on Machine Learning in Communications and Networking, and Elsevier Computer Networks.
\end{IEEEbiography}

\begin{IEEEbiography}[{\includegraphics[width=1in,height=1.25in,clip,keepaspectratio]{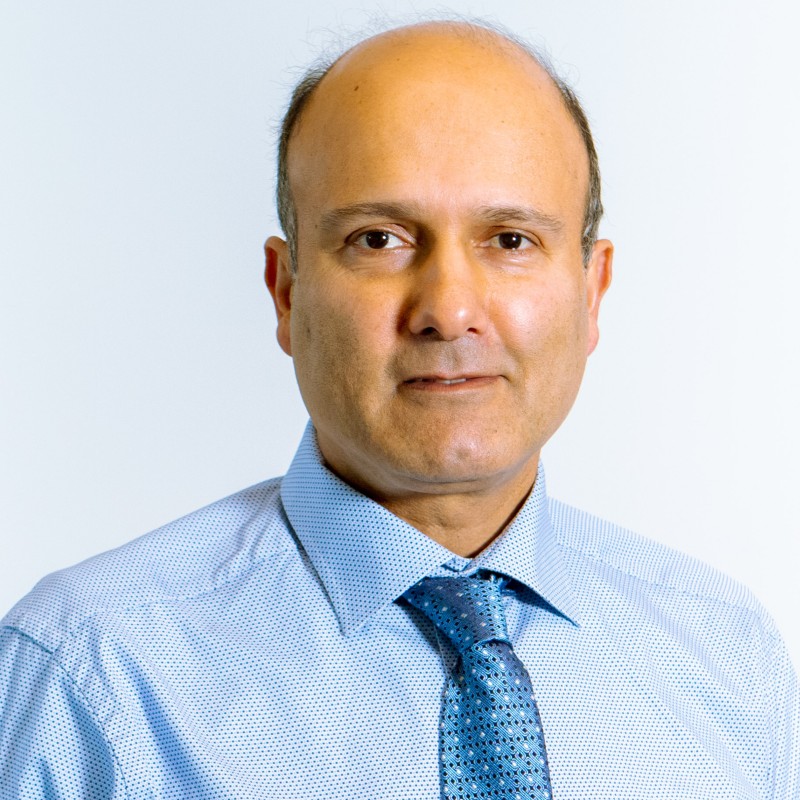}}]{Mohammad Shikh-Bahaei}
(Senior Member, IEEE) received the B.Sc. degree from the University of Tehran, Tehran, Iran, in 1992, the M.Sc. degree from the Sharif University of Technology, Tehran, in 1994, and the Ph.D. degree from King’s College London, London, U.K., in 2000.,He worked for two start-up companies, and National Semiconductor Corporation, Santa Clara, CA, USA (now part of Texas Instruments Inc.), on the design of third-generation mobile handsets, for which he has been awarded three U.S. patents as the Inventor and the Co-Inventor, respectively. In 2002, he joined King’s College London, as a Lecturer, where he is currently a Full Professor of Telecommunications with the Center for Telecommunications Research, Department of Engineering. He has authored/coauthored numerous journals and conference papers. He has been engaged in research in the area of wireless communications and signal processing for 25 years both in academic and industrial organizations. His research interests include learning-based resource allocation for multimedia applications over heterogeneous communication networks, full-duplex and UAV communications, and secure communication over wireless networks.,Dr. Shikh-Bahaei was a recipient of the Overall King’s College London Excellence in Supervisory Award in 2014. He was the Founder and the Chair of the Wireless Advanced (formerly, SPWC) Annual International Conference from 2003 to 2012.
\end{IEEEbiography}

\end{document}